\documentclass[11pt]{llncs}
\usepackage{latexsym}
\usepackage{graphicx}
\usepackage{enumerate}
\usepackage[ruled, vlined, linesnumbered]{algorithm2e}
\usepackage{booktabs}
\usepackage{amsmath}

\newcommand{\qedc}{\hfill $\Box$(Claim)}

\usepackage{color}
\newcommand{\commentout}[1]{}
\newcommand{\newekki}[1]{{\color{black} #1}} 

\newcommand{\tb}{{\sf tb}}

\newcommand{\tl}{{\sf tl}}
\newcommand{\cP}{{\cal P}}
\newcommand{\cO}{{\cal O}}
\newcommand{\cF}{{\cal F}}
\newcommand{\cT}{{\cal T}}
\newcommand{\cL}{{\cal L}}
\newcommand{\cI}{{\cal I}}
\newcommand{\pl}{{\sf pl}}
\newcommand{\pb}{{\sf pb}}
\newcommand{\pw}{{\sf pw}}
\newcommand{\pd}{{\sf pd}}

\newcommand{\ld}{{\sf ld}}
\newcommand{\bw}{{\sf bw}}

\usepackage[letterpaper]{geometry}
\newgeometry{left=2.5cm,right=2.5cm,top=2.5cm,bottom=2.5cm}

\usepackage{setspace}

\newcommand{\marginekki}[1]%
{%
  \setlength{\marginparwidth}{20mm}%
  \marginpar{%
    \begin{spacing}{0.5}
      \textsl{\tiny\color{green} #1}%
    \end{spacing}
  }%
}%

\begin{document}

\title{Line-distortion, Bandwidth and Path-length of a
  graph\thanks{Results of this paper were partially presented at the
    SWAT 2014 conference~\cite{SWATversion}.}}

\author{Feodor F. Dragan\inst{1} \and Ekkehard K\"ohler\inst{2} \and
  Arne Leitert\inst{1}}

\institute{
  Algorithmic Research Laboratory, Department of Computer Science, \\
  Kent State University, Kent, OH 44242, USA  \\
  \email{dragan@cs.kent.edu}, \email{aleitert@cs.kent.edu} \and
  Mathematisches Institut, Brandenburgische Technische Universit\"at Cottbus, \\
  D-03013 Cottbus, Germany \\
  \email{ekoehler@math.tu-cottbus.de} }

\maketitle

\begin{abstract}
  \newekki{For a graph $G=(V,E)$ the \emph{minimum line-distortion
      problem} asks for the minimum $k$ such that there is a mapping
    $f$ of the vertices into points of the line such that for each
    pair of vertices $x,y$ the distance on the line $|f(x) - f(y)|$
    can be bounded by the term $d_G(x, y)\leq |f(x)-f(y)|\leq k \,
    d_G(x, y)$, where $d_G(x, y)$ is the distance in the graph.  The
    \emph{minimum bandwidth problem} minimizes the term $\max_{uv\in
      E}|f(u)-f(v)|$, where $f$ is a mapping of the vertices of $G$
    into the integers \newekki{$\{1, \ldots, n\}$}.}

  We investigate the {minimum line-distortion} and the {minimum
    bandwidth} problems on unweighted graphs and their relations with
  the {\em minimum length} of a Robertson-Seymour's
  path-decomposition.  The {\em length} of a path-decomposition of a
  graph is the largest diameter of a bag in the decomposition.  The
  {\em path-length} of a graph is the minimum length over all its
  path-decompositions.  In particular, we show:

  \begin{itemize}
  \item if a graph $G$ can be embedded into the line with distortion
    $k$, then $G$ admits a Robertson-Seymour's path-decomposition with
    bags of diameter at most $k$ in $G$;
  \item for every class of graphs with path-length bounded by a
    constant, there exist an efficient constant-factor approximation
    algorithm for the minimum line-distortion problem and an efficient
    constant-factor approximation algorithm for the minimum bandwidth
    problem; 
  \item there is an efficient 2-approximation algorithm for computing
    the path-length of an arbitrary graph; 
  \item AT-free graphs and some intersection families of graphs have
    path-length at most 2; 
  \item for AT-free graphs, there exist a linear time 8-approximation
    algorithm for the minimum line-distortion problem and a linear
    time 4-approximation algorithm for the minimum bandwidth problem.
  \end{itemize}

  \medskip

\end{abstract}

\section{Introduction and previous work}\label{sec:intro}

Computing a minimum distortion embedding of a given $n$-vertex graph
$G$ into the line $\ell$ was recently identified as a fundamental
algorithmic problem with important applications in various areas of
computer science, like computer vision~\cite{TSL00}, as well as in
computational chemistry and biology (see~\cite{Indyk05,IndykM04}). It
asks, for a given graph
$G=(V,E)$, 
to find a mapping $f$ of vertices $V$ of $G$ into points of $\ell$
with minimum number $k$ such that $d_G(x, y)\leq |f(x)-f(y)|\leq k \,
d_G(x, y)$ for every $x,y\in V$. The parameter $k$ is called the {\em
  minimum line-distortion} of $G$ and denoted by $\ld(G)$. The
embedding $f$ is called {\em non-contractive} since $d_G(x, y)\leq
|f(x)-f(y)|$ for every $x,y\in V$.

In~\cite{BDG+05}, B\v{a}doiu et al.\ showed that this problem is hard
to approximate within a constant factor.  They gave an
exponential-time exact algorithm and a polynomial-time
$\cO(n^{1/2})$-approximation algorithm for arbitrary unweighted input
graphs, along with a polynomial-time $\cO(n^{1/3})$-approximation
algorithm for unweighted trees.  In another paper~\cite{BaChInSi}
B\v{a}doiu et al.\ showed that the problem is hard to approximate by a
factor $\cO(n^{1/12})$, even for weighted trees.  They also gave a
better polynomial-time approximation algorithm for general weighted
graphs, along with a polynomial-time algorithm that approximates the
minimum line-distortion $k$ embedding of a weighted tree by a factor
that is polynomial in $k$.

Fast exponential-time exact algorithms for computing the
line-distortion of a graph were proposed
in~\cite{FFLLRS09,FLS11}.  Fomin et al.\ in~\cite{FLS11} showed that a
minimum distortion embedding of an unweighted graph into the line can
be found in time $5^{n+o(n)}$.  Fellows et al.\ in~\cite{FFLLRS09}
gave an $\cO(nk^4(2k+1)^{2k})$ time algorithm that for an unweighted
graph $G$ and integer $k$ either constructs an embedding of $G$ into
the line with distortion at most $k$, or concludes that no such
embedding exists.  They extended their approach also to weighted
graphs obtaining an $\cO(nk^{4W} (2k + 1)^{2kW})$ time algorithm,
where $W$ is the largest edge weight.  Thus, the problem of minimum
distortion embedding of a given $n$-vertex graph $G$ into the line
$\ell$ is Fixed Parameter Tractable.

Recently, Heggernes et al.\ in~\cite{HM10,HMP11} initiated the study
of minimum distortion embeddings into the line of specific graph 
classes.  In particular, they gave polynomial-time algorithms for the
problem on bipartite permutation graphs and on threshold
graphs~\cite{HMP11}.  Furthermore, in~\cite{HM10}, Heggernes et al.\
showed that the problem of computing a minimum distortion embedding of
a given graph into the line remains NP-hard even when the input graph
is restricted to a bipartite, cobipartite, or split graph, implying
that it is NP-hard also on chordal, cocomparability, and AT-free
graphs.  They also gave polynomial-time constant-factor approximation
algorithms for split and cocomparability graphs.

Table~\ref{tbl:ExisPosLineDesResults} and
Table~\ref{tbl:ExisHardLineDesResults} summarise the results mentioned
above.

\begin{table}
  \scriptsize
  \centering
  \caption{Existing solutions for calculating the minimum line-distortion
    $\lambda$.}\label{tbl:ExisPosLineDesResults}
  \begin{tabular}{l  c c c}
    \toprule
    Graph Class~ & ~Solution Quality~ & ~Run Time~ & ~Source~  \\
    \midrule
    trees (unweighted) & $\cO(n^{1/3})$-approx. & polynomial & \cite{BDG+05} \\
    trees (weighted) & $\lambda^{\cO(1)}$-approx. & polynomial & \cite{BaChInSi} \\
    general (unweighted) & $\cO(n^{1/2})$-approx. & polynomial & \cite{BDG+05} \\
    & optimal & $5^{n+o(n)}$ & \cite{FLS11} \\
    & optimal & $~\cO(n \lambda^4 (2 \lambda+1)^{2\lambda})$~ & \cite{FFLLRS09} \\
    bipartite permutation~ & optimal & $\cO(n^2)$ & \cite{HMP11} \\
    threshold & optimal & linear & \cite{HMP11} \\
    split & 6-approx. & linear & \cite{HM10} \\
    cocomparability & 6-approx. & $\cO(n \log^2 n + m)$ & \cite{HM10} \\
    \bottomrule
  \end{tabular}
\end{table}
\begin{table}
  \scriptsize
  \centering
  \caption{Existing hardness results for calculating the minimum
    line-distortion.}\label{tbl:ExisHardLineDesResults}
  \begin{tabular}{l  cc}
    \toprule
    Graph Class & Result & ~Source~  \\
    \midrule
    general & $\cO(1)$-approximation is NP-hard & \cite{BDG+05} \\
    trees (weighted)~ & ~Hard to $\cO(n^{1/12})$-approximate~ & ~\cite{BaChInSi} \\
    bipartite & NP-hard & \cite{HM10} \\
    cobipartite & NP-hard & \cite{HM10} \\
    split & NP-hard & \cite{HM10} \\
    AT-free & NP-hard & \cite{HM10} \\
    cocomparability  & NP-hard & \cite{HM10} \\
    chordal & NP-hard & \cite{HM10} \\
    \bottomrule
  \end{tabular}
\end{table}

\newekki{The} minimum distortion embedding into the line may appear to
be closely related to the widely known and extensively studied graph
parameter {\em bandwidth}, denoted by $\bw(G)$.  The only difference
between the two parameters is that a minimum distortion embedding has
to be {\em non-contractive}, meaning that the distance in the
embedding between two vertices of the input graph has to be at least
their original distance, whereas there is no such restriction for
bandwidth.

Formally, given an unweighted graph $G=(V,E)$ on $n$ vertices,
consider a $1$-$1$ map $f$ of the vertices $V$ into integers in
\newekki{$\{1, \ldots, n\}$}; $f$ is called a {\em layout} of $G$. The
{\em bandwidth of layout} $f$ is defined as the maximum stretch of any
edge, i.e., $\bw(f) = \max_{uv \in E} |f(u) - f(v)|$.  The {\em
  bandwidth} of a graph is defined as the minimum possible bandwidth
achievable by any $1$-$1$ map (layout) $V \rightarrow \newekki{\{1,
  \ldots, n\}}$.  That is, $\bw(G) = \min_{f\colon V \rightarrow
  \newekki{\{1, \ldots, n\}}} \bw(f)$.

It is known that $\bw(G) \leq \ld(G)$ for every connected graph $G$
(see, e.g.,~\cite{HMP11}). 
However, the bandwidth and the minimum line-distortion of a graph can
be very different.  For example, it is common knowledge that a cycle
of length $n$ has bandwidth $2$, whereas its minimum line-distortion
is exactly $n - 1$~\cite{HMP11}.  Bandwidth is known to be one of the
hardest graph problems; it is NP-hard even for very simple graphs like
\emph{caterpillars of hair-length at most $3$} \newekki{(i.e., trees
  in which all the vertices are within distance 3 of a central
  path and all vertices of degree at least 3 are on the path)}~\cite{Monien86}, and it is hard to approximate by a constant
factor even for trees~\cite{BlKaWi97} and caterpillars with arbitrary
hair-lengths~\cite{DubeyaFeige2011}.  Polynomial-time algorithms for
the exact computation of bandwidth are known for very few graph
classes, including bipartite permutation graphs~\cite{HKM09} and
interval graphs~\cite{KV90,KS2002,Sprague94}.  Constant-factor
approximation algorithms are known for AT-free graphs~\cite{KKM99} and
convex bipartite graphs~\cite{ShTaUe2012}.  Recently,
in~\cite{GHKLMS2011} Golovach et al.\ showed also that the bandwidth
minimization problem is Fixed Parameter Tractable on AT-free graphs by
presenting an $n 2^{\cO(k \log k)}$ time algorithm. 
For general (unweighted) $n$-vertex graphs, the minimum bandwidth can
be approximated 
within a factor of
$\cO(\log^{3.5} n)$~\cite{Feige00}.  For $n$-vertex trees and chordal
graphs, the minimum bandwidth can be approximated 
within a factor of $\cO(\log^{2.5} n)$~\cite{Gupta01}.  For
$n$-vertex caterpillars with arbitrary hair-lengths, the minimum
bandwidth can be approximated to within a factor of $\cO(\log
n/\log\log n)$~\cite{FeigeTalwar2005}.

Table~\ref{tbl:ExisPosBwResults} and Table~\ref{tbl:ExisHardBwResults}
summarise the results mentioned above.

\begin{table}
  \scriptsize
  \centering
  \caption{Existing solutions for calculating the minimum
    bandwidth $k$.}\label{tbl:ExisPosBwResults}
  \begin{tabular}{l c c c}
    \toprule
    Graph Class & Solution Quality & Run Time & Source   \\
    \midrule
    caterpillars with hair-length 1 or 2 & optimal & $\cO(n \log n)$ & \cite{APSZ81} \\
    caterpillars with arbitrary hair-lengths~ & ~$\cO(\log n/\log\log n)$-approx.~ & ~polynomial~ & \cite{FeigeTalwar2005} \\
    general  & ~$\cO(\log^{3.5} n)$-approx.~ & polynomial & \cite{Feige00} \\
    chordal & ~$\cO(\log^{2.5} n)$-approx.~ & polynomial & \cite{Gupta01} \\
    AT-free & 2-approx. & $\cO(nm)$  & \cite{KKM99} \\
    & 4-approx. & ~$\cO(m + n \log n)$~  & \cite{KKM99} \\
    & optimal & $n2^{\cO(k\log k)}$ & \cite{GHKLMS2011} \\
    convex bipartite & 2-approx. & $\cO(n \log^2 n)$ & \cite{ShTaUe2012} \\
    & 4-approx. & $\cO(n)$ & \cite{ShTaUe2012} \\
    bipartite permutation & optimal & $\cO(n^4 \log n)$ & \cite{HKM09} \\
    interval & optimal & $\cO(n \log^2 n)$ & \cite{Sprague94} \\
    \bottomrule
  \end{tabular}
\end{table}

\begin{table}
  \scriptsize
  \centering
  \caption{Existing hardness results for calculating the minimum bandwidth.}\label{tbl:ExisHardBwResults}
  \begin{tabular}{l c c}
    \toprule
    Graph Class & Result & Source   \\
    \midrule
    trees & hard to approximate by a constant factor & \cite{BlKaWi97} \\
    caterpillars with arbitrary hair-lengths~ & ~hard to approximate by a constant factor~ & \cite{DubeyaFeige2011} \\
    caterpillars with hair-length at most 3 & NP-hard & \cite{Monien86} \\
    convex bipartite & NP-hard & \cite{ShTaUe2012} \\
    \bottomrule
  \end{tabular}
\end{table}

Our main tool in this paper is Robertson-Seymour's path-decomposition
and its length.  A \emph{path-decomposition}~\cite{RS83} of a graph
$G=(V,E)$ is a sequence of subsets $\{X_i: i\in I\}$
($I:=\{1,2,\dots,q\}$) of vertices of G, called {\em bags}, with three
properties:
\begin{enumerate}\label{def:path-decomposition}
\item $\bigcup_{i\in I}X_i=V$;
\item For each edge $uv\in E$, there is a bag $X_i$ such that $u,v \in
  X_i$;
\item For every three indices $i\le j\le k$, $X_i \cap X_k\subseteq
  X_j$.  Equivalently, the subsets containing any particular vertex
  form a contiguous subsequence of the whole sequence.
\end{enumerate}
We denote a path-decomposition $\{X_i: i\in I\}$ of a graph $G$ by
$\cP(G)$.

The \emph{width} of a path-decomposition $\cP(G)=\{X_i: i\in I\}$ is
$\max_{i\in I}|X_i|-1$. The {\em path-width} of a graph $G$, denoted
by $\pw(G)$, is the minimum width over all path-decompositions
$\cP(G)$ of $G$~\cite{RS83}. The caterpillars with hair-length at most
$1$ are exactly the graphs with
path-width~$1$~\newekki{\cite{proskurowski-telle-pathwidth}}.

Following~\cite{DoGa07} (where the notion of tree-length of a graph
was introduced), we define the {\em length} of a path-decomposition
$\cP(G)$ of a graph $G$ to be $\lambda := \max_{i \in I}\max_{u, v \in
  X_i}d_G(u, v)$ (i.e., each bag $X_i$ has diameter at most $\lambda$
in $G$).  The {\em path-length} of $G$, denoted by $\pl(G)$, is the
minimum length over all path-decompositions of $G$.  Interval graphs
(i.e., the intersection graphs of intervals on a line) are exactly the
graphs with path-length $1$; it is known (see,
e.g.,~\cite{Diestel00,FulGro1965,GilHof64,Gol-book}) that $G$ is an
interval graph if and only if $G$ has a path-decomposition with each
bag being a maximal clique of $G$.

Note that these two graph parameters (path-width and path-length) are
not related to each other.  For instance, a clique on $n$ vertices has
path-length $1$ and path-width $n-1$, whereas a cycle on $2n$ vertices
has path-width $2$ and path-length $n$.

Following~\cite{DraganK11}, where the notion of tree-breadth of a
graph was introduced, we define the breadth of a path-decomposition as
follows.  \newekki{Let $D_G(v_i,r)$ be the disk of radius $r$ around
  vertex $v_i$, more precisely, $D_G(v_i, r) = \{w \in V: d_G(v_i, w)
  \leq r\}$.  Then the \emph{breadth} of a path-decomposition $\cP(G)$
  of a graph $G$ is the minimum integer $r$ such that for every $i \in
  I$ there is a vertex $v_i \in V$ with $X_i \subseteq D_G(v_i, r)$
  (i.e., each bag $X_i$ can be covered by a disk $D_G(v_i, r)$ of
  radius at most $r$ in $G$).}  Note that vertex $v_i$ does not need
to belong to $X_i$.  The {\em path-breadth} of $G$, denoted by
$\pb(G)$, is the minimum breadth over all path-decompositions of $G$.
Evidently, for any graph $G$ with at least one edge, $1 \leq \pb(G)
\leq \pl(G) \leq 2 \, \pb(G)$ holds.  Hence, if one parameter is
bounded by a constant for a graph $G$ then the other parameter is
bounded for $G$ as well.


Recently, Robertson-Seymour's {tree-decompositions} with bags of
bounded radius proved to be very useful in designing an efficient
approximation algorithm for the problem of minimum stretch embedding
of an unweighted graph in to its spanning tree~\cite{DraganK11}.  The
decision version of the problem is the {\em tree $t$-spanner problem}
which asks, for a given graph $G=(V,E)$ and an integer $t$,
\newekki{whether} a spanning tree $T$ \newekki{of $G$} exists such
that $d_T(x, y) \leq t \, d_G(x, y)$ for every $x, y \in V$. It was
shown in~\cite{DraganK11} that: 
\begin{enumerate}[(1)]
\item if a graph $G$ can be embedded \newekki{into} a spanning tree
  with stretch $t$, then $G$ admits a Robertson-Seymour
  tree-decomposition with bags of radius at most $\lceil{t/2}\rceil$
  and diameter at most $t$ in $G$ (i.e., the tree-breadth $\tb(G)$ of
  $G$ is at most $\lceil{t/2}\rceil$ and the tree-length $\tl(G)$ of
  $G$ is at most $t$);
\item there is an efficient algorithm which constructs for an
  $n$-vertex unweighted graph $G$ with $\tb(G)\leq \rho$ a spanning
  tree with stretch at most $2\rho\log_2 n$.
\end{enumerate}
As a consequence, an efficient ($\log_2 n$)-approximation algorithm
was obtained for embedding an unweighted graph with minimum stretch
into its spanning tree~\cite{DraganK11}.

\subsection{Contribution of this paper} \label{ssec:our-res}

Motivated by~\cite{DraganK11}, in this paper, we investigate possible
connections between the line-distortion and the path-length
(path-breadth) of a graph.  We show that for every graph $G$, $\pl(G)
\leq \ld(G)$ and $\pb(G) \leq \lceil{\ld(G)/2}\rceil$ hold.
Furthermore, we demonstrate that for every class of graphs with
path-length bounded by a constant, there is an efficient
constant-factor approximation algorithm for the minimum
line-distortion problem.  As a consequence, every graph $G$ with
$\ld(G) = c$ can be embedded in polynomial time into the line with
distortion at most $\cO(c^2)$ (reproducing a result
from~\cite{BDG+05}).  Additionally, using the same technique, we show
that, for every class of graphs with path-length bounded by a
constant, there is an efficient constant-factor approximation
algorithm for the minimum bandwidth problem.

We also investigate 
(i) which particular graph classes have constant bounds on path-length
and (ii) how fast the path-length of an arbitrary graph can be
computed or sharply estimated. We present an efficient
$2$-approximation ($3$-approximation) algorithm for computing the
path-length (resp., the path-breadth) of a graph.  We show that
AT-free graphs and some \newekki{well-known} intersection families of
graphs have small path-length and path-breadth.  In particular, the
path-breadth of every permutation graph and every trapezoid graph is
$1$ and the path-length (and therefore, the path-breadth) of every
cocomparability graph and every AT-free graph is at most $2$.  Using
this and some additional structural properties, we give a linear time
$8$-approximation algorithm for the minimum line-distortion problem
and a linear time $4$-approximation algorithm for the minimum
bandwidth problem for AT-free graphs.

As a consequence of our results we obtain also that convex bipartite
graphs and caterpillars with hairs of bounded length admit constant
factor approximations of the minimum bandwidth and the minimum
line-distortion.  Furthermore, the minimum line-distortion problem and
the minimum bandwidth problem, are both NP-hard on bounded path-length
graphs.

\section{Preliminaries and metric properties of graphs with bounded
  path-length} \label{sec:met-prop}

All graphs occurring in this paper are connected, finite, unweighted,
undirected, loopless and without multiple edges. We call $G = (V, E)$
an {\em $n$-vertex $m$-edge graph} if $|V| = n$ and $|E| = m$.  In
this paper we consider only graphs with $n > 1$.  A {\em clique} is a
set of pairwise adjacent vertices of $G$.  By $G[S]$ we denote
\newekki{the} subgraph of $G$ induced by \newekki{the} vertices of $S
\subseteq V$. By $G\setminus S$ we denote the subgraph of $G$ induced by the vertices $V\setminus S$, i.e., the graph $G[V\setminus S]$.  For a vertex $v$ of $G$, the sets $N_G(v) = \{w \in V:
vw \in E\}$ and $N_G[v] = N_G(v) \cup \{v\}$ are called the \emph{open
  neighborhood} and the \emph{closed neighborhood} of $v$,
respectively.


In a graph $G$ the {\em length} of a path from a vertex $v$ to a
vertex $u$ is the number of edges in the path.  The {\em distance}
$d_G(u, v)$ between vertices $u$ and $v$ is the length of a shortest
path connecting $u$ and $v$ in $G$.  \newekki{For a set $S \subseteq
  V$ the {\em diameter} of $S$ in $G$} is $\max_{x, y \in S}d_G(x, y)$
and its {\em radius} in $G$ is $\min_{x \in V}\max_{y \in S}$ $d_G(x,
y)$ (in some papers \newekki{these terms} are called the {\em weak
  diameter} and the {\em weak radius} to indicate that the distances
are measured in $G$ not in $G[S]$).  The distance between a vertex $v$
and a set $S$ of $G$ is \newekki{given by} $d_G(v,S) = \min_{u \in S}
d_G(v, u)$.  The \emph{disk} of radius $k$ centered at vertex $v$ in
$G$ is the set of all vertices at distance at most $k$ from $v$, i.e.,
$D_G(v, k) = \{w \in V: d_G(v, w) \leq k\}.$
\commentout{ 
  Let also $G\setminus S$ be the graph $G[V\setminus S]$ (which is not
  necessarily connected). A set $S\subseteq V$ is called a {\em
    separator} of \muadstuff{a connected graph} $G$ if the graph
  $G[V\setminus S]$ has more than one connected component, and $S$ is
  called a {\em balanced separator} of $G$ if each connected component
  of $G[V\setminus S]$ has at most $|V|/2$ vertices. A set $C\subseteq
  V$ is called a {\em balanced clique-separator} of $G$ if $C$ is both
  a clique and a balanced separator of $G$.
}




The following result generalizes a characteristic property of the
famous class of AT-free graphs (see~\cite{CoOlSt97}).  An independent
set of three vertices such that each pair is joined by a path that
avoids the neighborhood of the third is called an {\em asteroidal
  triple}.  A graph $G$ is an {\em AT-free graph} if it does not
contain any asteroidal triples~\cite{CoOlSt97}.

\begin{proposition}\label{prop:AT-pl}
  Let $G$ be a graph with $\pl(G) \leq \lambda$. Then, for every three
  vertices $u, v, w$ of $G$ there is one vertex, say $v$, such that
  the disk of radius $\lambda$ centered at $v$ intercepts every path
  connecting $u$ and $w$, i.e., \newekki{after} the removal of the
  disk $D_G(v,\lambda)$ from $G$, \newekki{ $u$ and $w$ are not in the
    same connected component of $G \setminus D_G(v,\lambda)$.}
\end{proposition}

\begin{proof}
  Consider a path-decomposition $\cP(G) = \{X_i: i\in I\}$ of $G$ with
  length $\pl(G) \leq \lambda$. 
  Consider
  any three vertices $u, v, w$ of $G$.  If any two of them, say $u$
  \newekki{and} $v$, belong to same bag $X_i$ \newekki{of $\cP(G)$}
  then \newekki{the} disk $D_{G}(v,\lambda)$ contains vertex $u$ and
  hence intercepts every path of $G$ connecting vertices $u$ and $w$.
  Assume now, without loss of generality, that all bags containing
  vertex $u$ have smaller indexes in $I$ than all bags containing
  vertex $v$, and, in turn, all bags containing vertex $v$ have
  smaller indexes in $I$ than all bags containing vertex $w$.  Then,
  by properties of path-decompositions (see~\cite{Diestel00,RS83}),
  every $u,w$-path of $G$ \newekki{contains at least one} vertex in
  \newekki{each of the bags} between the bags containing $u$ and the
  bags containing $w$ in \newekki{the} sequence $\{X_i: i\in
  I\}$. 
  Hence, every
  bag $X_i$ containing $v$ intercepts every path connecting $u$ and
  $w$ in $G$.  Since $X_i$ is a subset of $D_{G}(v, \lambda)$, the
  proof is complete. \qed
\end{proof}

Since for every graph $G$, $\pl(G) \leq 2 \, \pb(G)$, the following
statement is also true.

\begin{corollary}\label{cor:AT-pb}
  Let $G$ be a graph with $\pb(G) \leq \rho$. Then, for every three
  vertices 
  of $G$, the disk of radius $2 \rho$ centered at one of them
  intercepts every path connecting the \newekki{other two} vertices.
\end{corollary}

We will also need the following property of graphs with path-length
$\lambda$.  A path $P$ of a graph $G$ is called {\em $k$-dominating
  path} of $G$ if every vertex $v$ of $G$ is at distance at most $k$
from a vertex of $P$, i.e., $d_G(v,P) \leq k$.  A pair of vertices
$x,y$ of $G$ is called a {\em $k$-dominating pair} if every path
between $x$ and $y$ is a $k$-dominating path of $G$. It is known that
every AT-free graph has a 1-dominating pair~\cite{CoOlSt97}.

\begin{corollary}\label{cor:DP-pb}
  Every graph $G$ with $\pl(G) \leq \lambda$ has a
  $\lambda$-dominating pair. 
\end{corollary}

\begin{proof}
  Consider a path-decomposition $\cP(G) = \{X_1,X_2,\dots,X_q\}$ of
  length $\pl(G) \leq \lambda$ of $G$.  Consider any two vertices $x
  \in X_1$ and $y \in X_q$ and a path $P$ between them in $G$.
  Necessarily, by properties of path-decompositions
  (see~\cite{Diestel00,RS83}), every path of $G$ connecting vertices
  $x$ and $y$ has a vertex in every bag of $\cP(G)$.  Hence, as each
  vertex $v$ of $G$ belongs to \newekki{some} bag \newekki{$X_{i}$} of
  $\cP(G)$, there is a vertex $u \in P$ with \newekki{$u \in X_i$ and
    thus} $d_G(v,u)\leq \lambda$.  \qed
\end{proof}

  
  \newekki{It is easy to see that a pair of vertices $x$ and $y$ is a $k$-dominating pair if and only  if, for every vertex
$w\in V\setminus(D_G(x,k)\cup D_G(y,k))$,  the disk $D_G(w,k)$ separates $x$ and $y$. Hence, a
   $k$-dominating pair, with minimum $k$, of an
  arbitrary graph $G = (V,E)$ with $n$ vertices and $m$ edges can be
  found in $\cO(n^3 \log n)$ time as follows.  As an outer loop use a
  binary search to find the minimum $k$.  Inside this loop determine
  for the corresponding $k$ whether there is a $k$-dominating pair by
  the following method.  For each vertex $v$ of $G$ determine the
  connected components of $G \setminus D_{G}(v,k)$, label each vertex
  $x$ in $G \setminus D_{G}(v,k)$ with its connected component, and
  put all these labels in a $n \times n$ matrix $M$, such that in
  $M(v,x)$ is the label of the connected component of vertex $x$ in $G
  \setminus D_{G}(v,k)$.  This matrix $M$ can easily be determined in
  $\cO(n (n+m))$ (for each vertex $v$ remove $D_{G}(v,k)$ and
  determine the connected components of the remaining graph by a Breadth-First-Search).
  Now a pair of vertices $x, y$ is a $k$-dominating pair if and only
  if the columns of $x$ and $y$ in $M$ have different labels in every
  row corresponding to a vertex $w$ with  $w\in V\setminus(D_G(x,k)\cup D_G(y,k))$.  Thus, an easy $\cO(n^3)$ algorithm for checking whether there
  is a $k$-dominating pair in $G$ is simply comparing for each pair of
  vertices the corresponding columns.}


\newekki{It is not very likely that there is a linear time algorithm
  to find a dominating pair, if it exists, since it is shown
  in~\cite{kratsch-spinrad-matrixmult} that finding a dominating pair
  is essentially as hard as finding a triangle in a graph.  Yet, since
  path-decompositions with small length are closely related to
  $k$-domination one can search for $k$-dominating pairs in dependence
  of the path-length of a graph.}  We do not know how to find {\sl in
  linear time} for an arbitrary graph $G$ a $k$-dominating pair with
$k\leq \pl(G)$.
\newekki{However,} we \newekki{can} prove the following weaker result
which will be useful in later sections.

\begin{proposition}\label{prop:dom-pair-linear}
  \newekki{Let $G$ be an arbitrary graph for which the path-length is
    not necessarily known.}  There is a \newekki{linear time algorithm
    that determines} a $k$-dominating pair of $G$ \newekki{such that
    $k \leq 2 \, \pl(G)$.}

\end{proposition}

\begin{proof}
  Let $\pl(G)= \lambda$. Consider a path-decomposition $\cP(G) =
  \{X_1, X_2, \ldots, X_q\}$ of $G$ of length $\lambda$.  Consider an
  arbitrary vertex $s$ of $G$ and, using a Breadth-First-Search
  $BFS(s,G)$ of $G$ started at $s$, find a vertex $x$ \newekki{of
    maximum distance} from $s$.  Use \newekki{a} second
  Breadth-First-Search $BFS(x,G)$ of $G$ \newekki{that is} started at
  $x$ to find a vertex $y$ \newekki{of maximum distance} from $x$.  We
  claim that $x,y$ is a $2 \, \lambda$-dominating pair of $G$.

  If there is a bag in $\cP(G)$ containing both $s$ and $x$, then
  $d_G(s,x) \leq \lambda$ and, by the choice of $x$, each vertex of
  $G$ is within distance at most $\lambda$ from $s$ and, hence, within
  distance at most $2 \, \lambda$ from $x$.  Evidently, in this case,
  $x,y$ is a $2 \, \lambda$-dominating pair of $G$.

  Assume now, without loss of generality, that $x \in X_i$ and $s \in
  X_l$ with $i < l$.  Consider an arbitrary vertex $v$ of $G$ that
  belongs to only bags with indexes smaller than $i$.  We show that
  $d_G(x, v) \leq 2 \, \lambda$.  As $X_i$ separates $v$ from $s$, a
  shortest path $P(s,v)$ of $G$ between $s$ and $v$ must have a vertex
  $u$ in $X_i$.  We have $d_G(s,x) \geq d_G(s,v) = d_G(s,u) +
  d_G(u,v)$ and, by the triangle inequality, $d_G(s,x) \leq d_G(s, u)
  + d_G(u, x)$.  Hence, $d_G(u, v) \leq d_G(u, x)$ and, since both $u$
  and $x$ belong to same bag $X_i$, $d_G(u, x) \leq \lambda$. That is,
  $d_G(x, v) \leq d_G(x, u) + d_G(u, v) \leq 2 d_G(u, x) \leq 2 \,
  \lambda$.

  If $d_G(x, y) \leq 2 \, \lambda$ then, by the choice of $y$, each
  vertex of $G$ is within distance at most $2 \, \lambda$ from $x$
  and, hence, $x, y$ is a $2 \, \lambda$-dominating pair of $G$.  So,
  assume that $d_G(x, y)> 2 \, \lambda$, i.e., every bag of $\cP(G)$
  that contains $y$ has index greater than $i$.  Consider a bag $X_j$
  containing $y$.  We have $i < j$. Repeating the arguments of the
  previous paragraph, we can show that $d_G(y, v) \leq 2 \, \lambda$
  for every vertex $v$ that belongs to bags with indexes greater than
  $j$.

  Consider now an arbitrary path $P$ of $G$ connecting vertices $x$
  and $y$.  By properties of path-decompositions
  (see~\cite{Diestel00,RS83}), $P$ has a vertex in every bag $X_h$ of
  $\cP(G)$ with $i \leq h \leq j$. Hence, for each vertex $v$ of $G$
  that belongs to a bag $X_h$ ($i \leq h \leq j$), there is a vertex
  $u \in P$ (in that bag $X_h$) such that $d_G(v, u) \leq \lambda$.
  As $d_G(v, x) \leq 2 \, \lambda$ for each vertex $v$ from $X_{i'}$
  with $i' < i$ and $d_G(v, y) \leq 2 \, \lambda$ for each vertex $v$
  from $X_{j'}$ with $j' > j$, we conclude that $P$ is a $2 \,
  \lambda$-dominating path of $G$.~\qed
\end{proof}

In Algorithm~\ref{algo:2plDomPath}, we formalize the method in the
proof above to calculate a $k$-dominating shortest path with $k\leq
2\pl(G)$ in linear time.

\begin{algorithm}
  \KwIn{A graph $G$.}%
  \KwOut{A $k$-dominating shortest path with $k\leq 2\pl(G)$.}%

  \smallskip

  Select an arbitrary vertex $s$.

  Find a vertex $x$ for which the distan\newekki{ce} to $s$ is maximal.

  Find a vertex $y$ for which the distan\newekki{ce} to $x$ is maximal.

  Output a shortest path from $x$ to $y$.

  \caption{Finding a $k$-dominating shortest path of $G$ with $k \leq
    2\pl(G)$.}\label{algo:2plDomPath}
\end{algorithm}

The following proposition further strengthens the connections between
small path-length graphs and AT-free graphs.  Recall that the
$k$-power of a graph $G = (V,E)$ is a graph $G^{k} = (V, E')$ such
that for every $x,y \in V$ ($x \neq y$), $xy \in E'$ if and only if
$d_G(x, y) \leq k$.

\begin{proposition}
  For a graph $G$ with $\pl(G) \leq \lambda$, $G^{2\,\lambda-1}$ is an
  AT-free graph.
\end{proposition}

\begin{proof}
  Let $\cP(G)$ be a path-decomposition of length $\pl(G)\leq \lambda$
  of $G$ \newekki{and let $G^{2\,\lambda-1} = (V, E')$ be the
    $(2\,\lambda-1)$-power of $G$}.  Consider three arbitrary distinct vertices
  $a$, $b$ and $c$ of $G$.  If for two of those vertices there is a
  bag~$B \in \cP(G)$ containing both, then $a$, $b$ and $c$ cannot be
  an asteroidal triple in $G^{2\,\lambda-1}$, since they do not form
  an independent set in $G^{2\,\lambda-1}$.  Now assume that
  \newekki{no bag contains more than one of $a$, $b$, and
    $c$.}  
  Without loss of generality, we can assume that $b$ is in a bag~$B_b$
  between the bags containing $a$ and the bags containing $c$.  Assume
  \newekki{that} there is a path from $a$ to $c$ in $G^{2\,\lambda-1}$
  avoiding $B_b$.  
  Then, there is an edge $uv \in E'\setminus E$ such that the bags
  containing $u$ and the bags containing $v$ are separated by $B_b$
  \newekki{in $G$}.  Since $uv \in E'$, there is a path of length at
  most $2\,\lambda-1$ from $u$ to $v$ in $G$, and \newekki{again, by
    properties of path-decompositions (see~\cite{Diestel00,RS83}),}
  this path must contain a vertex~$w \in B_b$.  Without loss of
  generality, let $d_G(u, w) \leq d_G(v,w)$.  Since $d_G(u, w) \leq
  d_G(u, v)/2 \leq \lambda-1$, $d_G(b, u) \leq d_G(b, w) + d_G(u, w)
  \leq 2\,\lambda - 1$, i.e., $u \in D_G(b, 2\,\lambda-1)$.  Thus,
  each path from $a$ to $c$ of $G^{2\,\lambda - 1}$ intersects $D_G(b,
  2\,\lambda-1)$, \newekki{implying that $a, b, c$ cannot form an
    asteroidal triple in $G^{2\,\lambda - 1}$}. 
  \qed 
\end{proof}

\begin{corollary}
  If $\pb(G) \leq \rho$, then $G^{4\rho-1}$ is AT-free.
\end{corollary}

A subset of vertices of a graph is called {\em connected} if the
subgraph induced by those vertices is connected.  We say that two
connected sets $S_1$, $S_2$ of a graph $G$ {\em see each other} if
they have a common vertex or there is an edge in $G$ with one end in
$S_1$ and the other end in $S_2$.  A family of connected subsets of
$G$ is called a {\em bramble} if every two sets of the family see each
other.  We say that a bramble $\cF = \{S_1, \dots, S_h\}$ of $G$ is
$k$-dominated by a vertex $v$ of $G$ if in every set $S_i$ of $\cF$
there is a vertex $u_i \in S_i$ with $d_G(v, u_i) \leq k$.

\begin{proposition}\label{prop:bramble}
  For a graph $G$ with $\pb(G) \leq \rho$, every bramble of $G$ is
  $\rho$-dominated by a vertex.
\end{proposition}

\begin{proof}
  Let $\cP(G) = \{X_1, X_2, \dots, X_q\}$ be a path-decomposition of
  breadth $\pb(G) \leq \rho$ of $G$.  Consider an arbitrary connected
  set $S$ of $G$.  We claim that the bags of $\cP(G)$ containing
  vertices of $S$ form a continuous subsequence $\cI(S)$ in $\{X_1,
  X_2, \dots, X_q\}$.  Assume, by induction on the cardinality of the
  set \newekki{$S$}, that the statement is true for \newekki{the}
  connected set $S' := S\setminus \{v\}$, where $v$ is some vertex of
  $S$ \newekki{(obviously, there is always such a vertex in $S$)}.
  Let $\cI(S')$ be the corresponding subsequence.  Since $S$ is
  connected, there must exist a vertex $u$ in $S'$ such that $uv \in
  E(G)$.  By properties \newekki{2 and 3} of \newekki{the}
  path-decomposition \newekki{(see definition on
    page~\pageref{prop:bramble})}, all bags containing vertex $v$ form
  a continuous subsequence $\cI(v)$ in $\{X_1, X_2, \dots, X_q\}$ and
  there is a bag in $\cP(G)$ which contains both vertices $u$ and $v$.
  Then, necessarily, all bags containing vertices of $S$ form a
  continuous subsequence in $\{X_1, X_2, \dots, X_q\}$; it is the
  union of \newekki{the} two continuous subsequences $\cI(S')$ and
  $\cI(v)$ sharing a common bag.

  \newekki{Now let} $\cF = \{S_1, \dots, S_h\}$ be an arbitrary
  bramble of $G$.  For every set $S_i$, the bags of $\cP(G)$
  containing vertices of $S_i$ form a continuous subsequence
  $\cI(S_i)$ in $\{X_1, X_2, \dots, X_q\}$.  Since each two sets of
  $\cF$ see each other, there must exist a bag in $\cP(G)$ that
  contains a vertex from $S_i$ and a vertex from $S_j$.  So, for every
  $i,j \in \{1, \dots, h\}$, the subsequences $\cI(S_i)$ and
  $\cI(S_j)$ overlap at least on one bag.  By the Helly property for
  intervals of a line (i.e., every family of pairwise intersecting
  intervals has a common intersection), there must exist a bag $B$ in
  $\cP(G)$ which has a vertex from each set $S_i$ ($i \in \{1, \dots,
  h\}$).  Let $v$ be a vertex of $G$ such that $B \subseteq D_G(v,
  \rho)$.  Then, $v$ necessarily $\rho$-dominates the bramble $\cF$.
  \qed
\end{proof}

The following result can be viewed as an analog of the classical Helly
property for disks.

\begin{corollary}\label{cor:helly}
  Let $G$ be a graph with $\pb(G) \leq \rho$, \newekki{let} $S$ be a
  subset of vertices of $G$\newekki{,} and \newekki{let} $r: S
  \rightarrow \mathbf{N}$ be a radius function defined on $S$ such
  that the disks of the family $\cF = \{D_G(x, r(x)) :x \in S\}$
  pairwise intersect.  Then the disks $\{D_G(x, r(x) + \rho) : x \in
  S\}$ have a nonempty common intersection.
\end{corollary}

\begin{proof}
  Since the family $\cF = \{D_G(x, r(x)) :x \in S\}$ is a bramble of
  $G$, a vertex $v$ of $G$ $\rho$-dominating the bramble $\cF$ belongs
  to all disks $\{D_G(x, r(x) + \rho) : x \in S\}$. \qed
\end{proof}


\section{Bandwidth of graphs with bounded path-length}
\label{sec:bw-vs-pl}

In this section we show that there is an efficient algorithm that for
any graph $G$ with $\pl(G)=\lambda$ produces a layout $f$ with
bandwidth at most 
$\cO(\lambda) \bw(G)$. Moreover, this statement is true even for all
graphs with $\lambda$-dominating shortest paths. Recall that a
shortest path $P$ of a graph $G$ is a {\em $k$-dominating shortest
  path} of $G$ if every vertex $v$ of $G$ is at distance at most $k$
from a vertex of $P$, i.e., $d_G(v,P)\leq k$. %

We will need the following standard ``local density'' lemma.

\begin{lemma}[\cite{lecturenotes}]\label{lem:Discbandwidth}
  For each vertex $v \in V$ of an arbitrary graph $G$ and each
  positive integer $r$, $$\frac{|D_G(v, r)| - 1}{2r} \leq \bw(G).$$
\end{lemma}

The main result of this section is the following.

\begin{proposition}\label{prop:DP-bw}
  Every graph $G$ with a $k$-dominating shortest path has a layout $f$
  with bandwidth at most $(4k+2) \, \bw(G)$.
  If a $k$-dominating shortest path of $G$ is given in advance,
  then such a layout $f$ can be found in linear time.
\end{proposition}

\begin{proof}
  Let $P=(x_0,x_1,\dots,x_i,\dots, x_j,\dots, x_q)$ be a
  $k$-dominating shortest path of $G$.  Consider a
  Breadth-First-Search-tree $T_P$ of $G$ started from path $P$, i.e.,
  \newekki{$T_{P}$ is the} $BFS(P,G)$-tree of $G$.  For each vertex
  $x_i$ of $P$, let $X_i$ be the set of vertices of $G$ that are
  located in the branch of $T_P$ that is rooted at $x_i$ (see
  Fig.~\ref{fig:bw-vs-pb1}(a) for an illustration). We have $x_i\in
  X_i$.  Since $P$ $k$-dominates $G$, we have $d_G(v,x_i)\leq k$ for
  every $i\in \{1,\dots,q\}$ and every $v\in X_i$. Now create a layout
  $f$ of $G$ by placing \newekki{all the} vertices of $X_i$ before all
  vertices of $X_j$, if $i < j$, and by placing \newekki{the} vertices
  within each $X_i$ in an arbitrary order (see
  Fig.~\ref{fig:bw-vs-pb1}(b) for an illustration).

  \begin{figure}[htb]
    \begin{center} \vspace*{-25mm}
      \begin{minipage}[b]{16cm}
        \begin{center} 
          \hspace*{-28mm}
          \includegraphics[height=16cm]{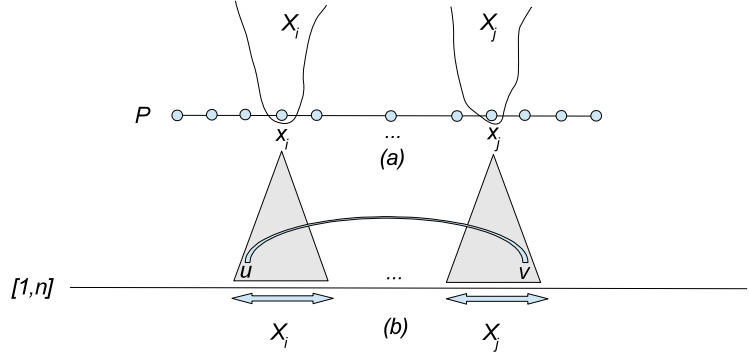}
        \end{center} \vspace*{-94mm}
        \caption{\label{fig:bw-vs-pb1} Illustration to the proof of
          Proposition \ref{prop:DP-bw}.  } %
      \end{minipage}
    \end{center}
  \end{figure}

  We claim that this layout $f$ has bandwidth at most
  $(4k+2)\,\bw(G)$.  Consider any edge $uv$ of $G$ and assume $u \in
  X_i$ and $v \in X_j$ ($i \le j)$.  For this edge $uv$ we have $f(v)
  - f(u) \leq |\bigcup_{\newekki{\ell} = i}^jX_\newekki{\ell}| - 1$.
  We \newekki{also} know that $d_P(x_i, x_j) = j-i \leq 2k + 1$, since
  $P$ is a shortest path of $G$ and $d_P(x_i, x_j) = d_G(x_i, x_j)
  \leq d_G(x_i, u) + 1 + d_G(x_j, v) \le 2k + 1$.  Consider vertex
  $x_c$ of $P$ with $c = i + \lfloor{(j - i)/2}\rfloor$, i.e., a
  middle vertex of the subpath of $P$ between $x_i$ and $x_j$.
  Consider an arbitrary vertex $w$ in $X_\newekki{\ell}$, $i \le
  \newekki{\ell} \le j$. Since
  $d_G(x_c, w) \leq d_G(x_c, x_\newekki{\ell}) + d_G(x_\newekki{\ell},
  w)$, $d_G(x_c, x_\newekki{\ell}) \le \lceil{2k + 1}\rceil/2$ and
  $d_G(x_\newekki{\ell}, w) \le k$, we get $d_G(x_c, w) \leq 2k + 1$.
  In other words, disk $D_G(x_c, 2k + 1)$ contains all vertices of
  $\bigcup_{\newekki{\ell} = i}^j X_\newekki{\ell}$.  Applying
  Lemma~\ref{lem:Discbandwidth} to $|D_G(x_c, 2k+1)| \geq
  |\bigcup_{\newekki{\ell} = i}^j X_\newekki{\ell}|$, we conclude
  $f(v) - f(u) \leq |\bigcup_{\newekki{\ell} = i}^j X_\newekki{\ell}|
  - 1 \leq |D_G(x_c, 2k + 1)| - 1 \leq 2(2k+1)\,\bw(G) = (4k +
  2)\,\bw(G)$. \qed
\end{proof}

Proposition~\ref{prop:DP-bw}, Corollary~\ref{cor:DP-pb}, and
Proposition~\ref{prop:dom-pair-linear} imply.

\begin{corollary}\label{cor:appr-bw-pl}
  For every $n$-vertex $m$-edge graph $G$, a layout with bandwidth at
  most $(4\,\pl(G)+2)\,\bw(G)$ can be found in $\cO(n^2 m)$ time and a
  layout with bandwidth at most $(8\,\pl(G)+2)\,\bw(G)$ can be found
  in $\cO(n + m)$ time.
\end{corollary}

\begin{proof}
  For an $n$-vertex $m$-edge graph $G$, a $k$-dominating shortest path
  with $k \leq \pl(G)$ can be found in $\cO(n^2 m)$ time in the
  following way (see Algorithm~\ref{algo:plDomPath}). Iterate over all
  vertex pairs of $G$.  For each vertex pair \newekki{$x,y$} pick a
  shortest \newekki{$x,y$-path $P$} and run $BFS(P,G)$ to find a most
  distant vertex $v_P$ from $P$.  Finally, report that path $P$ for
  which $d_G(v_P,P)$ is minimum. By Corollary~\ref{cor:DP-pb}, this
  minimum is at most $\pl(G)$.


  \begin{algorithm}
    \KwIn{A graph $G$.}%
    \KwOut{A $k$-dominating shortest path with $k\leq \pl(G)$.}%

    \smallskip

    \ForEach {vertex pair $x, y$}%
    {

      Find a shortest path $P_{xy}$ from $x$ to $y$.

      Determine $k(x, y) := \max_{v \in V} d_G(v, P_{xy})$.

    }

    Output a path $P_{xy}$ for which $k(x,y)$ is minimal.

    \caption{Finding a $k$-dominating shortest path of $G$ with $k
      \leq \pl(G)$.}\label{algo:plDomPath}
  \end{algorithm}

  Alternatively, one can use the proof of
  Proposition~\ref{prop:dom-pair-linear} to find in linear time a
  $2\pl(G)$-dominating pair $x, y$ of $G$.  Then, any shortest path of
  $G$ between $x$ and $y$ is a $2\,\pl(G)$-dominating path of $G$ (see
  Algorithm~\ref{algo:2plDomPath}).

  The entire method for computing a required layout is given in
  Algorithm~\ref{algo:4kBwApprox}. Its runtime and approximation
  \newekki{ratio} depend on the algorithm to calculate a
  $k$-dominating shortest path.  \qed
\end{proof}

\begin{algorithm}
  \KwIn{A graph $G=(V,E)$.}%
  \KwOut{A layout $f$.}%

  \smallskip

  Find a $k$-dominating shortest path $P = (x_0, x_1, \ldots, x_q)$
  using Algorithm~\ref{algo:2plDomPath} or
  Algorithm~\ref{algo:plDomPath}.

  Partition $V$ into sets $X_0, X_1, \ldots, X_q$ using
  \newekki{a} $BFS(P,G)$-tree of $G$ (see the proof of
  Proposition~\ref{prop:DP-bw}).

  Create a layout $f$ of $G$ by placing \newekki{all the} vertices of
  $X_i$ before all vertices of $X_j$, if $i < j$, and by placing
  vertices within each $X_i$ in an arbitrary order.

  Output $f$.

  \caption{An $\cO(k)$-approximation algorithm for \newekki{computing}
    the minimum bandwidth of a graph \newekki{using} a $k$-dominating
    shortest path.}\label{algo:4kBwApprox}
\end{algorithm}

Thus, we have the following interesting conclusion.

\begin{theorem}\label{th:bw-pl}
  For every class of graphs with path-length bounded by a constant,
  there is an efficient constant-factor approximation algorithm for
  the minimum bandwidth problem.
\end{theorem}

The above results did not require a path-decomposition of length
$\pl(G)$ of a graph $G$ \newekki{as input}; we \newekki{also} avoided
the construction of such a 
path-decomposition of $G$ and just relied on the existence of a
$k$-dominating shortest path in $G$.  If\newekki{, however,} a
path-decomposition with length $\lambda$ of a graph $G$ is given in
advance together with $G$, then a better approximation ratio for the
minimum bandwidth problem on $G$ can be achieved.

\begin{proposition}\label{prop:PD-bw}
  If a graph $G$ is given together with a path-decomposition of $G$ of
  length $\lambda$, then a layout $f$ with bandwidth at most $\lambda
  \, \bw(G)$ can be found in $\cO(n^2 + n \log^2 n)$ time.
\end{proposition}

\begin{proof}
  Let $\cP(G) = \{X_i: i \in I\}$ be a path-decomposition of length
  $\lambda$ of $G = (V, E)$.  We form a new graph $G^+=(V,E^+)$
  \newekki{from $G$} by adding an edge between a pair of vertices $u,
  v \in V$ if and only if $u$ and $v$ belong to a common bag in
  $\cP(G)$.  From this construction, we conclude that $G$ is a
  subgraph of $G^+$ and $G^+$ is a subgraph of $G^{\lambda}$.  It is a
  well-known fact (see,
  e.g.,~\cite{Bra-Book,Diestel00,FulGro1965,Gol-book}) that $G^+$ is
  an interval graph and $\cP(G) = \{X_i: i \in I\}$ gives a
  path-decomposition of $G^+$ (with $\{X_i: i \in I\}$ being cliques
  of $G^+$).  In~\cite{Sprague94}, an $\cO(n \log^2 n)$ time algorithm
  to compute a minimum bandwidth layout of an $n$-vertex interval
  graph is given.  Let $f$ be an optimal layout produced by that
  algorithm for our interval graph $G^+$.  We claim that this layout
  $f$\newekki{, when} considered for $G$\newekki{,} has bandwidth at
  most $\lambda \, \bw(G)$.  Indeed, following~\cite{KKM99}, we have
  $\max_{uv \in E}|f(u) - f(v)| \leq \max_{uv \in E^+}|f(u) - f(v)| =
  \bw(G^+) \leq \bw(G^{\lambda}) \leq \lambda \,\bw(G)$.  Clearly,
  raising a graph to the $\lambda$th power can only increase its
  bandwidth by a factor of $\lambda$.~\qed
\end{proof}

We formalize the method described above in
Algorithm~\ref{algo:plBandwidthApprox}.

\begin{algorithm}
  \KwIn{A graph $G$ with a path-decomposition $\cP(G) = \{ X_1,
    \ldots, X_q \}$.}%
  \KwOut{A layout $f$.}%

  \smallskip

  Create a new graph $G^+=(V,E^+)$ by adding an edge between
  \newekki{each} pair of vertices $u,v \in V$ if and only if $u$ and
  $v$ belong to a common bag in $\cP(G)$.

  Compute the minimum bandwidth layout $f$ of \newekki{the} interval
  graph $G^+$ by using an optimal $\cO(n \log^2 n)$ time algorithm
  from~\cite{Sprague94}.

  Output $f$.

  \caption{A $\lambda$-approximation algorithm for \newekki{computing}
    the minimum bandwidth for a graph with path-length
    $\lambda$.}\label{algo:plBandwidthApprox}
\end{algorithm}

We do not know how hard \newekki{it is} for an arbitrary graph to
construct its path-decomposition with minimum length.  We suspect that
\newekki{this} is an NP-hard problem as the problem to check
\newekki{whether} a graph has tree-length at most $\lambda$ is
NP-complete for every fixed $\lambda \geq 2$~\cite{Daniel10}.  In
Section~\ref{sec:appr-pl-pb} we show that a factor 2 approximation of
the path-length of an arbitrary $n$-vertex graph can be computed in
$\cO(n^3)$ time.  This \newekki{implies} in particular that, for an
arbitrary $n$-vertex graph $G$, a layout with bandwidth at most
$2\,\pl(G)\,\bw(G)$ can be found in $\cO(n^3)$ total time.

Additionally, in Section~\ref{sec:spec-classes} we show that the
path-breadth of every permutation graph and every trapezoid graph
is~$1$ and the path-length (and therefore, the path-breadth) of every
cocomparability graph and every AT-free graph is at
most~$2$. 
In Section~\ref{sec:appr-ld-AT}, using some additional structural
properties of AT-free graphs, we give a linear time 4-approximation
algorithm for the minimum bandwidth problem \newekki{for AT-free
  graphs}.  This result reproduces an approximation result
from~\cite{KKM99} with a better run-time.  Note that the class of
AT-free graphs properly contains all permutation graphs, trapezoid
graphs and cocomparability graphs; definitions of these graph classes
are given in Section~\ref{sec:spec-classes}.

\section{Path-length and line-distortion} \label{sec:ld-vs-pl}

In this section, we first show that the line-distortion of a graph
gives an upper bound on its path-length and then demonstrate that if
the path-length of a graph $G$ is bounded by a constant then there is
an efficient constant-factor approximation algorithm for the minimum
line-distortion problem on $G$.

\subsection{Bound on line-distortion implies bound on path-length}
\label{ssec:ld-bounds-pl}

In this subsection we show that the path-length of an arbitrary graph
never exceeds its line-distortion.  The following inequalities are
true.

\begin{proposition}\label{prop:ld-vs-pb}
  For an arbitrary graph $G$, $\pl(G) \leq \ld(G)$, $\pw(G) \leq
  \ld(G)$ and $\pb(G) \leq \lceil{\ld(G)/2}\rceil$.
\end{proposition}

\begin{proof}
  It is known (see, e.g.,~\cite{HMP11}) that every connected graph $G
  = (V, E)$ has a minimum distortion embedding $f$ into the line
  $\ell$ (called a {\em canonic} embedding) such that $|f(x) - f(y)| =
  d_G(x, y)$ for every two vertices \newekki{$x,y$} of $G$ that are
  placed next to each other in $\ell$ by $f$.  Assume, in what follows,
  that $f$ is such a canonic embedding and let $k := \ld(G)$.

  Consider the following path-decomposition of $G$ created from $f$.
  For each vertex $v$, form a bag $B_v$ consisting of all vertices of
  $G$ which are placed by $f$ in the interval $[f(v), f(v) + k]$ of
  the line $\ell$.  Order these bags with respect to the left ends of
  the corresponding intervals.  Evidently, for every vertex $v \in V$,
  $v \in B_v$, i.e., each vertex belongs to a bag.  More generally, a
  vertex $u$ belongs to a bag $B_v$ if and only if $f(v) \leq f(u)
  \leq f(v) + k$.  Since $\ld(G) = k$, for every edge $uv$ of $G$,
  $|f(u) - f(v)| \le k$ holds.  Hence, both ends of edge $uv$ belong
  either to bag $B_u$ (if $f(u) < f(v)$) or to bag $B_v$ (if $f(v) <
  f(u)$).  \newekki{Now consider} three bags $B_a$, $B_b$, and $B_c$
  with $f(a) < f(b) < f(c)$ and a vertex $v$ of $G$ that belongs to
  $B_a$ and $B_c$.  We have $f(a) < f(b) < f(c) \leq f(v) \leq f(a) +
  k < f(b) + k$.  Hence, necessarily, $v$ belongs to $B_b$ as well.

  It remains to show that each bag $B_v$, $v \in V$, has in $G$
  diameter at most $k$, radius at most $\lceil{k/2}\rceil$ and
  cardinality at most $k + 1$.  Indeed, for any two vertices $x, y \in
  B_v$, we have $|f(x) - f(y)| \leq k$, i.e., $d_G(x, y) \le |f(x) -
  f(y)| \leq k$.  Furthermore, any interval $[f(v), f(v) + k]$ (of
  length $k$) can have at most $k + 1$ vertices of $G$ as the distance
  between any two vertices placed by $f$ to this interval is at least
  1 ($|f(x) - f(y)| \ge d_G(x, y) \ge 1$).  Thus, $|B_v| \leq k + 1$
  for every $v \in V$.

  \newekki{Now consider} the point $p_v := f(v) + \lfloor{k/2}\rfloor$
  in the interval $[f(v), f(v) + k]$ of $\ell$.  Assume, without loss
  of generality, that $p_v$ is between $f(x)$ and $f(y)$, the images
  of two vertices $x$ and $y$ of $G$ placed next to each other in
  $\ell$ by $f$.  Let $f(x) \leq p_v < f(y)$ (see
  Fig.~\ref{fig:ld-vs-pb} for an illustration).  Since $f$ is a
  canonic embedding \newekki{of $G$}, there must exist 
  a vertex $c$ on a shortest path
  between $x$ and $y$ such that $d_G(x, c) = p_v - f(x)$ and $d_G(c,
  y) = f(y) - p_v = d_G(x, y) - d_G(x, c)$. We claim that for every
  vertex $w \in B_v$, $d_G(c, w) \le \lceil{k/2}\rceil$ holds.  Assume
  $f(w) \ge f(y)$ (the case when $f(w) \le f(x)$ is similar).  Then,
  we have $d_G(c, w) \le d_G(c, y) + d_G(y, w) \leq (f(y) - p_v) +
  (f(w) - f(y)) = f(w) - p_v = f(w) - f(v) - \lfloor{k/2}\rfloor
  \leq k - \lfloor{k/2}\rfloor \leq
  \lceil{k/2}\rceil$. \qed 
\end{proof}

\begin{figure}[htb]
  \begin{center} \vspace*{-29mm}
    \begin{minipage}[b]{16cm}
      \begin{center} 
        \hspace*{-28mm}
        \includegraphics[height=16cm]{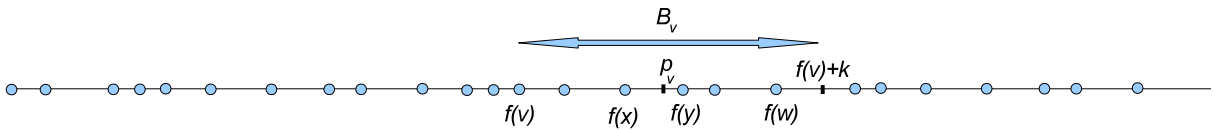}
      \end{center} \vspace*{-120mm}
      \caption{\label{fig:ld-vs-pb} Illustration to the proof of
        Proposition \ref{prop:ld-vs-pb}.  } %
    \end{minipage}
  \end{center}
\end{figure}

It should be noted that the difference between the path-length and the
line-distortion of a graph can be very large.  A complete graph $K_n$
on $n$ vertices has path-length $1$, whereas the line-distortion of
$K_n$ is $n-1$.  Note also that the bandwidth and the path-length of a
graph do not bound each other.  The bandwidth of $K_n$ is $n-1$ while
its path-length is $1$.  On the other hand, the path-length of cycle
$C_{2n}$ is $n$ while its bandwidth is $2$.

\subsection{Line-distortion of graphs with bounded
  path-length}\label{ssec:pb-vs-ld}

In this subsection we show that there is an efficient algorithm that
for any graph $G$ with $\pl(G) = \lambda$ produces 
an embedding $f$ of $G$ into the line $\ell$ with distortion at most
$(12 \, \lambda+7)\,\ld(G)$.  Again, this statement is true even for
all graphs with $\lambda$-dominating shortest paths.


We will need the following auxiliary lemma from~\cite{{BDG+05}}.  We
reformulate it slightly.  Recall that a subset of vertices of a graph
is called {\em connected} if the subgraph induced by those vertices is
connected.

\begin{lemma}[\cite{BDG+05}]\label{lem:small-coon-set}
  Any connected subset $S \subseteq V$ of a graph $G = (V, E)$ can be
  embedded into the line with distortion at most $2|S| - 1$ in time
  $\cO(|V| + |E|)$.  In particular, there is a mapping $f$, computable
  in $\cO(|V| + |E|)$ time, of \newekki{the} vertices from $S$ into
  points of the line such that $d_G(x, y) \leq |f(x) - f(y)| \leq 2|S|
  - 1$ for every $x,y \in S$.
\end{lemma}

The main result of this subsection is the following.

\begin{proposition}\label{prop:DP-ld}
  Every graph $G$ with a $k$-dominating shortest path admits an
  embedding $f$ of $G$ into the line with distortion at most $(8k +
  4)\,\ld(G) + (2k)^2 + 2k + 1$. 
  If a $k$-dominating shortest path of $G$ is given in advance,
  then such an embedding $f$ can be found in linear time.
\end{proposition}

\begin{proof}
  Like in the proof of Proposition~\ref{prop:DP-bw}, consider a
  $k$-dominating shortest path $P = (x_0, x_1, \dots, x_i, \dots, x_j,
  \dots, x_q)$ of $G$ and identify by $BFS(P,G)$ the sets $X_i$, $i
  \in \{1,\dots,q\}$.  We had $d_G(v, x_i) \leq k$ for every $i \in
  \{1, \dots, q\}$ and every $v \in X_i$.  It is clear also that each
  $X_i$ is a connected subset of $G$.  Similar to~\cite{BDG+05}, we
  define an embedding $f$ of $G$ into the line $\ell$ by placing
  \newekki{all the} vertices of $X_i$ before all vertices of $X_j$, if
  $i < j$, and by placing vertices within each $X_i$ in accordance
  with the embedding mentioned in Lemma~\ref{lem:small-coon-set}.
  Also, for each $i \in \{1, \dots, q-1\}$, leave a space of length
  $2k + 1$ between the interval of $\ell$ spanning the vertices of
  $X_i$ and the interval spanning the vertices of $X_{i + 1}$. See
  Fig.~\ref{fig:ld-vs-pb1} for an illustration.

  \begin{figure}[htb]
    \begin{center} \vspace*{-25mm}
      \begin{minipage}[b]{16cm}
        \begin{center} 
          \hspace*{-28mm}
          \includegraphics[height=16cm]{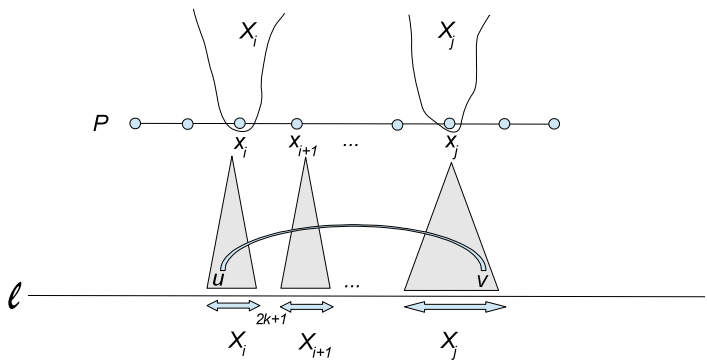}
        \end{center} \vspace*{-94mm}
        \caption{\label{fig:ld-vs-pb1} Illustration to the proof of
          Proposition \ref{prop:DP-ld}.  } %
      \end{minipage}
    \end{center}
  \end{figure}

  We claim that $f$ is a (non-contractive) embedding with distortion
  at most $(8k + 4) \, \ld(G) + (2k)^2 + 2k + 1$.  It is sufficient to
  show that $d_G(x, y) \leq |f(x) - f(y)|$ for every two vertices of
  $G$ that are placed next to each other in $\ell$ by $f$ and that
  $|f(v) - f(u)| \leq (8k + 4) \, \ld(G) + (2k)^2 + 2k + 1$ for every
  edge $uv$ of $G$ (see, e.g.,~\cite{BDG+05,HMP11}).

  From Lemma~\ref{lem:small-coon-set}, we know that $d_G(x, y) \leq
  |f(x) - f(y)| \leq 2|X_h| - 1$ for every $x,y \in X_h$ and $h \in
  \{1, 2, \dots, q\}$.  Additionally, for every $x \in X_i$ and $y \in
  X_{i + 1}$ ($i \in \{1, 2, \dots, q-1\}$), we have $d_G(x, y) \leq
  d_G(x, x_i) + 1 + d_G(y, x_{i + 1}) \le 2k + 1 \leq |f(y) - f(x)|$
  (as a space of length $2k + 1$ is left between the interval of
  $\ell$ spanning the vertices of $X_i$ and the interval spanning the
  vertices of $X_{i + 1}$).  Hence, $f$ is non-contractive.

  Consider now an arbitrary edge $uv$ of $G$ and assume $u \in X_i$
  and $v \in X_j$ ($i \le j)$.  For this edge $uv$ (by Lemma~\ref{lem:small-coon-set}) we have %
  $f(v) - f(u) \leq \sum_{h = i}^j (2|X_h| - 1 + 2k + 1) - 2k - 1 = 
  2|\bigcup_{h = i}^j X_h| + 2k(j - i + 1) - 2k - 1 = 2|\bigcup_{h =
    i}^j X_h| + 2k(j - i) - 1$. 
    Recall that $d_P(x_i,
  x_j) = j - i \leq 2k + 1$, since $P$ is a shortest path of $G$ and
  $d_P(x_i, x_j) = d_G(x_i, x_j) \leq d_G(x_i, u) + 1 + d_G(x_j, v)
  \le 2k + 1$.  Hence, $f(v) - f(u) \leq 2|\bigcup_{h = i}^j X_h| +
  2k(2k + 1) - 1$.

  As in the proof of Proposition~\ref{prop:DP-bw}, $|\bigcup_{h = i}^j
  X_h| - 1 \leq (4k + 2) \, \bw(G)$.  As $\bw(G) \le \ld(G)$ for every
  graph $G$ (see, e.g.,~\cite{HMP11}), we get $f(v) - f(u) \leq
  2|\bigcup_{h = i}^j X_h| + 2k(2k + 1) - 1 \leq 2(4k + 2) \, \bw(G) +
  2k(2k + 1) + 1 \leq (8k + 4) \, \ld(G) + 2k(2k + 1) + 1$. \qed
\end{proof}


Proposition~\ref{prop:DP-ld}, Corollary~\ref{cor:DP-pb} and
Proposition~\ref{prop:dom-pair-linear} imply.

\begin{corollary}\label{cor:appr-ld-pl}
  For every $n$-vertex $m$-edge graph $G$, an embedding into the line
  with distortion at most $(12 \, \pl(G) + 7) \, \ld(G)$ can be found
  in $\cO(n^2 m)$ time and with distortion at most $(24 \, \pl(G) + 7)
  \, \ld(G)$ can be found in $\cO(n + m)$ time.
\end{corollary}

\begin{proof}
  See the proof of Corollary~\ref{cor:appr-bw-pl} and note that, by
  Proposition~\ref{prop:ld-vs-pb}, $\pl(G) \leq \ld(G)$.  Hence, the
  distortion established in Proposition~\ref{prop:DP-ld}
  becomes 
  $\leq (8 \, \pl(G) + 4) \, \ld(G) + 2(2 \, \pl(G) + 1) \, \ld(G) + 1
  \leq (12 \, \pl(G) + 7) \, \ld(G)$, if we use a $\pl(G)$-dominating
  shortest path, and becomes $\leq (24 \, \pl(G) + 7) \, \ld(G)$, if
  we use a $2 \, \pl(G)$-dominating shortest path.
  Algorithm~\ref{algo:8kLineDesApprox} covers both cases.  Its runtime
  and approximation \newekki{ratio} depend on the algorithm to
  calculate a $k$-dominating shortest path.  \qed
\end{proof}

\begin{algorithm}
  \KwIn{A graph $G$.}%
  \KwOut{An embedding $f$ of $G$ into the line $\ell$.}%

  \smallskip

  Find a $k$-dominating shortest path $P=(x_0, x_1, \ldots, x_q)$
  using Algorithm~\ref{algo:2plDomPath} or
  Algorithm~\ref{algo:plDomPath}.

  Partition $V$ into sets $X_0, X_1, \ldots, X_q$ using \newekki{a}
  $BFS(P,G)$-tree of $G$ (see the proof of
  Proposition~\ref{prop:DP-ld}).

  Create an embedding $f$ of $G$ into the line $\ell$ by placing
  \newekki{all the} vertices of $X_i$ before all vertices of $X_j$, if
  $i < j$, and
  by placing vertices within each $X_i$ in accordance with the
  embedding mentioned in Lemma~\ref{lem:small-coon-set}.  Also, for
  each $i \in \{1, \dots, q - 1\}$, leave a space of length $2k + 1$
  between the interval of $\ell$ spanning the vertices of $X_i$ and
  the interval spanning the vertices of $X_{i + 1}$.

  Output $f$.

  \caption{An $\cO(\lambda)$-approximation algorithm for
    \newekki{computing} the minimum line-distortion of a graph $G$
    with $\pl(G) \leq \lambda$.}\label{algo:8kLineDesApprox}
\end{algorithm}

Thus, we have the following interesting conclusion.

\begin{theorem}\label{th:ld-pl}
  For every class of graphs with path-length bounded by a constant,
  there is an efficient constant-factor approximation algorithm for
  the minimum line-distortion problem.
\end{theorem}

Using inequality $\pl(G) \leq \ld(G)$ in
Corollary~\ref{cor:appr-ld-pl} once more, we reproduce a result
of~\cite{BDG+05}.

\begin{corollary}[\cite{BDG+05}]\label{cor:MIT}
  For every graph $G$ with $\ld(G) = c$, an embedding into the line
  with distortion at most $\cO(c^2)$ can be found in polynomial time.
\end{corollary}

It should be noted that, since the difference between the path-length
and the line-distortion of a graph can be very large (close to $n$),
the result in Corollary~\ref{cor:appr-ld-pl} seems to be stronger.

\medskip

Theorem~\ref{th:bw-pl} and Theorem~\ref{th:ld-pl} stress the
importance of \newekki{investigating the questions} (i) \newekki{for
  which} particular graph classes \newekki{there is a} constant bound
on \newekki{the} path-length and of (ii) how fast \newekki{can} the
path-length of an arbitrary graph be computed or sharply estimated.
In the next two sections we address some of those questions.

\section{Constant-factor approximations of path-length and
  path-breadth}\label{sec:appr-pl-pb}

Let $G = (V, E)$ be an arbitrary graph and \newekki{let} $s$ be
\newekki{an} arbitrary vertex \newekki{of $G$}.  A \emph{layering}
${\cal L}(s, G)$ of $G$ with respect to a start vertex $s$ is the
decomposition of $V$ into layers $L_i = \{u \in V: d_G(s, u) = i\}$,
$i = 0, 1, \dots, q$.  \newekki{For an integer $i \geq 1$ and a vertex
  $v \in L_i$ denote by $N_G^{\downarrow}(v) = N_G(v)\cap L_{i-1}$ the
  neighborhood of $v$ in the previous layer
  $L_{i-1}$.} 
  We can get a
path-decomposition of $G$ by adding to each layer $L_i$ ($i > 0$) all
vertices from layer $L_{i - 1}$ that have a neighbor in
$L_i$\newekki{, in particular, let} $L^+_i := L_i \cup (\bigcup_{v \in
  L_i} \newekki{N_G^{\downarrow}(v)})$.
Clearly, the sequence $\{L^+_1, \dots, L^+_q\}$ is a
path-decomposition of $G$ and can be constructed in $\cO(|E|)$ total
time.  We call this path-decomposition an {\em extended layering} of
$G$ and denote it by $\cL^+(s, G)$.

\newekki{As shown in the next theorem,} this type of
path-decomposition has length at most twice as large as the
path-length of the graph.

\begin{theorem}\label{th:approx-tl}
  For every graph $G$ with $\pl(G) = \lambda$ there is a vertex $s$
  such that the length of the extended layering $\cL^+(s, G)$ of $G$
  is at most $2 \, \lambda$. In particular, a factor $2$ approximation
  of the path-length of an arbitrary $n$-vertex 
  graph can be computed in 
  $\cO(n^3)$ total time.
\end{theorem}

\begin{proof}
  Consider a path-decomposition $\cP(G) = \{X_1, X_2, \dots, X_p\}$ of
  length $\pl(G) = \lambda$ of $G$.  Let $s$ be an arbitrary vertex
  from $X_1$.  Consider the layering ${\cal L}(s, G)$ of $G$ with
  respect to $s$ where $L_i = \{u \in V: d_G(s, u) = i\}$\newekki{,}
  $(i = 0, 1, \dots, q$).  Let $x$ and $y$ be two arbitrary vertices
  from $L_i$ ($i \in \{1, \dots, q\}$) and $x'$ and $y'$ be arbitrary
  vertices from $L_{i - 1}$ with $xx', yy' \in E$.  We will show that
  $\max\{d_G(x, y), d_G(x, y'), d_G(x', y)\} \leq 2 \, \lambda$.  By
  induction on $i$, we may assume that $d_G(y', x') \le 2 \, \lambda$
  as $x',y' \in L_{i - 1}$.

  If there is a bag in $\cP(G)$ containing both vertices $x$ and $y$,
  then $d_G(x, y) \le \lambda$ and therefore $d_G(x, y') \le \lambda +
  1 \leq 2 \, \lambda$, $d_G(y, x') \le \lambda + 1 \leq 2 \, \lambda$.
  Assume now that all bags containing $x$ are earlier in $\cP(G) =
  \{X_1, X_2, \dots, X_p\}$ than the bags containing $y$.  Let $B$ be
  a bag of $\cP(G)$ containing both ends of edge $xx'$ (such a bag
  necessarily exists by properties of path-decompositions).  By the
  position of this bag $B$ in $\cP(G)$ and the fact that $s \in X_1$,
  any shortest path connecting $s$ with $y$ must have a vertex in $B$.
  Let $w$ be a vertex of $B$ that is on a shortest path of $G$
  connecting vertices $s$ and $y$ and containing edge $yy'$.  Such a
  shortest path must exist because of the structure of the layering
  ${\cal L}(s, G)$ that starts at $s$ and puts $y'$ and $y$ in
  consecutive layers.  \newekki{Since $x, x', w \in B$} we have
  $\max\{d_G(x, w), d_G(x', w)\} \leq \lambda$.  If $w = y'$ then we
  are done; $\max\{d_G(x, y), d_G(x, y'), d_G(x', y)\} \leq \lambda +
  1 \leq 2 \, \lambda$.  So, assume that $w \neq y'$.  Since $d_G(x, s) =
  d_G(s, y) = i$ (by the layering) and $d_G(x, w) \leq \lambda$, we
  must have $d_G(w, y') + 1 = d_G(w, y) = d_G(s, y) - d_G(s, w) =
  d_G(s, x) - d_G(s, w) \leq d_G(w, x) \leq \lambda$.  Hence, $d_G(y,
  x) \leq d_G(y, w) + d_G(w, x) \leq 2 \, \lambda$, $d_G(y, x') \leq
  d_G(y, w) + d_G(w, x') \leq 2 \, \lambda$ and $d_G(y', x) \leq d_G(y',
  w) + d_G(w, x) \leq 2 \, \lambda - 1$.

  We conclude that the distance between any two vertices \newekki{of}
  $L_i^+$ is at most $2 \, \lambda$, that is, the length of \newekki{the}
  tree decomposition $\cL^+(s, G)$ of $G$ is at most $2 \, \lambda$. \qed
\end{proof}

Algorithm~\ref{algo:2plApprox} formalizes the method described above.

\begin{algorithm}
  \KwIn{A graph $G=(V,E)$.}%
  \KwOut{A path-decomposition for $G$.}%

  \smallskip

  Calculate distances $d_G(u, v)$ for all vertices $u,v \in V$.

  \ForEach {$s \in V$}%
  {%
    Calculate a decomposition ${\cal L}^+(s\newekki{, G}) = \{
    L_0^+(s), L_1^+(s), \ldots \}$ with $L_i(s) = \{v \in V : d_G(s,
    v) = i \}$ and $L_i^+(s) = L_i(s) \cup \{v \in L_{i - 1}(s) :
    N_G(v) \cap L_i(s) \neq \emptyset \} $.

    Determine the length $l(s)$ of ${\cal L}^+(s\newekki{, G})$%
  }

  Output a decomposition ${\cal L}^+(s\newekki{, G})$ for which $l(s)$
  is minimal.

  \caption{A 2-approximation algorithm for \newekki{computing} the
    path-length of a graph.}\label{algo:2plApprox}
\end{algorithm}

\begin{theorem}\label{th:approx-tb}
  For every graph $G$ with $\pb(G) = \rho$ there is a vertex $s$ such
  that the breadth of the extended layering $\cL^+(s, G)$ of $G$ is at
  most $3 \rho$.  In particular, a factor $3$ approximation of the
  path-breadth of an arbitrary $n$-vertex 
  graph can be computed in 
  $\cO(n^3)$ total time.
\end{theorem}

\begin{proof}
  Since $\pl(G) \leq 2 \, \pb(G)$, by Theorem~\ref{th:approx-tl},
  there is a vertex $s$ in $G$ such that the length of extended
  layering $\cL^+(s,G) = \{L_1^+, \dots, L_q^+\}$ of $G$ is at most $4
  \rho$.  Consider a bag $L_i^+$ of $\cL^+(s, G)$ and a family $\cF =
  \{D_G(x, 2 \rho) :x \in L_i^+\}$ of disks of $G$.  Since $d_G(u, v)
  \leq 4 \rho$ for every pair $u,v \in L_i^+$, the disks of $\cF$
  pairwise intersect.  Hence, by Corollary~\ref{cor:helly}, the disks
  $\{D_G(x, 3 \rho) : x \in L_i^+\}$ have a nonempty common
  intersection.  A vertex $w$ from that common intersection has all
  vertices of $L_i^+$ within distance at most $3 \rho$.  That is, for
  each $i \in \{1, \dots, q\}$ there is a vertex $w_i$ with $L_i^+
  \subseteq D_G(w_i, 3 \rho)$. \qed
\end{proof}

Combining Theorem~\ref{th:approx-tl} and Proposition~\ref{prop:PD-bw},
we obtain the following result.

\begin{theorem}\label{th:bw-pl-int}
  For every $n$-vertex graph $G$, a layout $f$ with bandwidth at most
  $2 \, \pl(G) \, \bw(G)$ can be found in $\cO(n^3)$ total time.
\end{theorem}

\section{Bounds on path-length and path-breadth for special graph
  classes}\label{sec:spec-classes}

The class of AT-free graphs contains many intersection families of
graphs, among them interval graphs, permutation graphs, trapezoid
graphs and cocomparability graphs.  These three families of graphs can
be defined as follows~\cite{Bra-Book,Gol-book}.  Consider a line in
the plane and $n$ intervals on this line.  The intersection graph of
such a set of intervals is called an {\em interval graph}.  Consider
two parallel lines (upper and lower) in the plane.  Assume that each
line contains $n$ points, labeled $1$ to $n$.  \newekki{Each} two
points with the same label define a segment with that label.  The
intersection graph of such a set of segments between two parallel
lines is called a {\em permutation graph}.  Assume now that each
\newekki{of the two parallel} line\newekki{s} contains $n$ intervals,
labeled $1$ to $n$, and each two intervals with the same label define
a trapezoid with that label (a trapezoid can degenerate to a triangle
or to a segment).  The intersection graph of such a set of trapezoids
between two parallel lines is called a {\em trapezoid graph}.
Clearly, every permutation graph is a trapezoid graph, but not vice
versa.  The class of cocomparability graphs (which contains all
interval graphs and all trapezoid graphs as subclasses) can be defined
as the intersection graphs of continuous function
diagrams\newekki{~\cite{golumbic-rotem-urrutia}}, but for this paper
it is more convenient to define them via the existence of a special
vertex ordering.  A graph $G$ is a {\em cocomparability graph} if it
admits a vertex ordering $\sigma = [v_1, v_2, \dots, v_n]$, called a
{\em cocomparability ordering}, such that for any $i < j < k$, if
$v_i$ is adjacent to $v_k$ then $v_j$ must be adjacent to at least one
of $v_i$, $v_k$.  According to~\cite{McCSp}, such an ordering of a
cocomparability graph can be constructed in linear time.  Note also
that, given a permutation graph $G$, a {\em permutation model} (i.e.,
a set of segments between two parallel lines, defining $G$) can be
found in linear time~\cite{McCSp}; a {\em trapezoid model} for an
$n$-vertex trapezoid graph can be found in $\cO(n^2)$
time~\cite{MaSp}.

In this section we show that the path-breadth of every permutation
graph and every trapezoid graph is $1$ and the path-length (and
therefore, the path-breadth) of every cocomparability graph and every
AT-free graph is bounded by $2$.

\begin{proposition}\label{prop:perm}
  If $G$ is a permutation graph, then $\pb(G) = 1$ and, therefore,
  $\pl(G) \leq 2$.  Furthermore, a path-decomposition of $G$ with
  breadth $1$ can be computed in linear time.
\end{proposition}

\begin{proof}
  We assume that a permutation model of $G$ is given in advance (if
  not, we can compute one for $G$ in linear time~\cite{McCSp}).  That
  is, each vertex $v$ of $G$ is associated with a segment $s(v)$ such
  that $uv \in E$ if and only if segments $s(v)$ and $s(u)$ intersect.
  In what follows, ``u.p.'' and ``l.p.'' refer to a vertex's point on
  the upper and lower, respectively, line of the permutation model.

  First we compute a\newekki{n (inclusion)} maximal independent set
  $M$ of $G$ in linear time as follows.  Put in $M$ (which is
  initially empty) a vertex $x_1$ whose u.p.\ is leftmost.  For each
  $i \geq 2$, select a vertex $x_i$ whose u.p.\ is leftmost among all
  vertices whose segments do not intersect $s(x_1), \dots, s(x_{i -
    1})$ (in fact, it is enough to check intersection with $s(x_{i -
    1})$ only).  If such a vertex exists, put it in $M$ and continue.
  If no such vertex exists, $M = \{x_1, \dots, x_k\}$ has been
  constructed.

  Now, we claim that $\{N_G[x_1], \dots, N_G[x_k]\}$ is a
  path-decomposition of $G$ with breadth $1$ and, hence, with length
  at most $2$.  Clearly, each vertex of $G$ is in some bag since every
  vertex not in $M$ is adjacent to a vertex in $M$, by \newekki{the}
  maximality of $M$.  Consider an arbitrary edge $uv$ of $G$.  Assume
  that neither $u$ nor $v$ is in $M$ and that the u.p.\ of $u$ is to
  the left of the u.p.\ of $v$.  Necessarily, the l.p.\ of $v$ is to
  the left of the l.p.\ of $u$, as segments $s(v)$ and $s(u)$
  intersect.  Assume that the u.p.\ of $u$ is between the u.p.s of
  $x_i$ and $x_{i + 1}$.  From the construction of $M$, $s(u)$ and
  $s(x_i)$ must intersect, i.e., the l.p.\ of $u$ is to the left of
  the l.p.\ of $x_i$.  But then, since the l.p.\ of $v$ is to the left
  of the l.p.\ of $x_i$, segments $s(v)$ and $s(x_i)$ must intersect,
  too.  Thus, edge $uv$ is in bag $N_G[x_i]$.

  To show that all bags containing any particular vertex form a
  contiguous subsequence of the sequence $N_G[x_1], \dots, N_G[x_k]$,
  consider an arbitrary vertex $v$ of $G$ and let $v \in N_G[x_i] \cap
  N_G[x_j]$ for $i < j$.  Consider an arbitrary bag $N_G[x_l]$ with $i
  < l < j$.  We know that vertices $x_i, x_l, x_j \in M$ are pairwise
  non-adjacent.  Furthermore, segment $s(v)$ intersects segments
  $s(x_i)$ and $s(x_j)$.  As segment $s(x_l)$ is between $s(x_i)$ and
  $s(x_j)$, necessarily, $s(v)$ intersects $s(x_l)$ as well.~\qed
\end{proof}

\begin{proposition}\label{prop:trap}
  If $G$ is an $n$-vertex trapezoid graph, then $\pb(G) = 1$ and,
  therefore, $\pl(G) \leq 2$.  Furthermore, a path-decomposition of
  $G$ with breadth $1$ can be computed in $\cO(n^2)$ time.
\end{proposition}

\begin{proof}
  We will show that every trapezoid graph $G$ is a minor of a
  permutation graph.  Recall that a graph $G$ is called a {\em minor}
  of a graph $H$ if $G$ can be formed from $H$ by deleting edges and
  vertices and by contracting edges.

  First, we compute in $\cO(n^2)$ time a trapezoid model for
  $G$~\cite{MaSp}.  Then, we replace each trapezoid $\cT_i$ in this
  model with its two diagonals obtaining a permutation model with $2
  n$ vertices.  Let $H$ be the permutation graph of this permutation
  model.  It is easy to see that two trapezoids $\cT_1$ and $\cT_2$
  intersect if and only if a diagonal of $\cT_1$ and a diagonal of
  $\cT_2$ intersect.

  Now, $G$ can be obtained back from $H$ by a series of $n$ edge
  contractions; for each trapezoid $\cT_i$, contract the edge of $H$
  that corresponds to two diagonals of $\cT_i$.

  Since contracting edges does not increase the path-breadth
  (see~\cite{DraganK11}), we get $\pb(G) = \pb(H) = 1$ by
  Proposition~\ref{prop:perm}.  Any path-decomposition of $H$ with
  breadth $1$ is a path-decomposition of $G$ with breadth $1$.~\qed
\end{proof}

\newekki{%
  For the proof of the next result we will need a special vertex
  ordering $\sigma: V \rightarrow \{1, \dots, n\}$ produced by a
  so-called {\em Lexicographic-Breadth-First-Search} (LBFS for short)
  which is a refinement of a standard {\em Breadth-First-Search}
  (BFS).  LBFS$(s,G)$ starts at some start vertex $s$, orders the
  vertices of a graph $G$ by assigning numbers from $n$ to $1$ to the
  vertices in the order as they are discovered by the following search
  process.  Each vertex $v$ has a \emph{label} consisting of a
  (revers) ordered list of the \emph{numbers} of those neighbors of
  $v$ that were already visited by the LBFS; initially this label is
  empty.  LBFS starts with some vertex $s$, assigns number $n$ to $s$,
  and adds this number to the end of the label of all un-numbered
  neighbors of $s$.  Then, in each step, LBFS selects the un-numbered
  vertex $v$ with the lexicographically largest label, assigns the
  next available number $k$ to $v$, and adds this number to the end of
  the labels of all un-numbered neighbors of $v$.  An ordering
  $\sigma$ of the vertex set of a graph generated by LBFS$(s,G)$ is
  called a LBFS$(s,G)$-ordering.  Note that the closer a vertex is to
  $s$ in $G$ the larger its number is in $\sigma$.
  It is known that a LBFS-ordering of an arbitrary graph can be
  generated in linear time~\cite{Gol-book}.  }%

\begin{proposition}\label{prop:AT}
  If $G$ is an $n$-vertex AT-free graph, then $\pb(G) \leq \pl(G) \leq
  2$.  Furthermore, a path-decomposition of $G$ with length at most
  $2$ can be computed in $\cO(n^2)$ time.
\end{proposition}

\begin{proof}
  Let $s$ be an arbitrary vertex of $G$ and $x$ be a vertex last
  visited (numbered $1$) by an LBFS$(s,G)$.  Let $\sigma$ be an
  LBFS$(x,G)$-ordering of vertices of $G$.  Clearly, $\sigma$ can be
  generated in linear time as one needs only $2$ scans of LBFS to do
  that.  The following useful result was proven in~\cite{COS-SICOMP}.

  \medskip%
  \noindent%
  {\em Claim 1:~\cite{COS-SICOMP} %
    For every vertex $y$ of an AT-free graph $G$, the pair $x,y$ is a
    $1$-dominating pair of the subgraph $G_{\geq\sigma(y)}$ of $G$
    induced by vertices $\{z \in V: \sigma(y) \leq \sigma(z) \leq
    \sigma(x) = n\}$.  In particular, $x$ and the vertex last visited
    by LBFS$(x,G)$ constitute a $1$-dominating pair of $G$.%
  }
  \medskip


  Let now $\mathcal{L}\newekki{(u,G)} = \{L_0, \ldots, L_k\}$ with
  $L_i = \{ u \in V : d_G(u,x) = i \}$ be a layering of $G$ produced
  by LBFS$(x,G)$. 

  \medskip%
  \noindent%
  {\em Claim 2: %
    For every integer $i \geq 1$ and every two non-adjacent vertices
    $u,v \in L_i$ of an AT-free graph $G$, $\sigma(v) < \sigma(u)$
    implies $N_G^{\downarrow}(v) \subseteq N_G^{\downarrow}(u)$.
    In particular, $d_G(v, u) \leq 2$ holds for every $u,v\in L_i$ and
    every $i$.%
  }

  \begin{proof}[of Claim~2]
    Consider an arbitrary neighbor $w \in L_{i - 1}$ of $v$ and a
    shortest path $P$ from $v$ to $x$ in $G$ containing $w$.  Since
    $\sigma(v) < \sigma(u)$, by Claim 1,
    path $P$ must dominate vertex $u$.  Since $u$ and $v$ are not
    adjacent, $u$ is in $L_i$ and all vertices of $P \setminus
    \{v,w\}$ belong to layers $L_j$ with $j < i - 1$, vertex $u$ must
    be adjacent to $w$.\qedc
  \end{proof}

  We can transform \newekki{an} AT-free graph $G = (V, E)$ into an
  interval graph $G^+ = (V, E^+)$ by applying the following two
  operations:
  \begin{enumerate}[(1)]
  \item ({\it make layers complete graphs}) In each layer $L_i$, make
    every two vertices $u,v \in L_i$ adjacent to each other in $G^+$;
  \item ({\it make down-neighborhoods of adjacent vertices of a layer
      comparable, too}) For each $i$ and every edge $uv$ of $G$ with
    $u,v \in L_i$ and $\sigma(v) < \sigma(u)$, make every $w \in
    N_G^{\downarrow}(v)$ adjacent to $u$ in $G^+$.
  \end{enumerate}

  Clearly, for every edge $uw$ of $G^+$ added by operation~(2),
  $d_G(u, w) \leq 2$ holds.  Also, for every edge $uv$ of $G^+$ added
  by operation~(1), $d_G(u, v) \leq 2$ holds by Claim~2.
  Thus, we have.

  \medskip%
  \noindent%
  {\em Claim 3: %
    $G^+$ is a subgraph of $G^2$.  }
  \medskip

  Next we show that $G^+$ is an interval graph.

  \medskip%
  \noindent%
  {\em Claim 4: %
    $G^+$ is an interval graph.
  }

  \begin{proof}[of Claim~4]
    It is known~\cite{Olariu91} that a graph is an interval graph if
    and only if its vertices admit an {\sl interval ordering}, i.e.,
    an ordering $\tau: V \rightarrow \newekki{\{1, \dots, n\}}$ such
    that for every choice of vertices $a, b, c$ with $\tau(a) <
    \tau(b) < \tau(c)$, $ac \in E$ implies $bc \in E$.  We show here
    that the LBFS$(x,G)$-ordering $\sigma$ of $G$ is an interval
    ordering of $G^+$.  Recall that, for every $v \in L_i$ and $u \in
    L_j$ with $i > j$, \newekki{it holds that} $\sigma(v) < \sigma(u)$
    (as $\sigma$ is a\newekki{n} LBFS-ordering).  Consider
    \newekki{three} arbitrary vertices $a,b,c$ of $G$ and assume that
    $\sigma(a) < \sigma(b) < \sigma(c)$ and $ac \in E^+$.  Assume also
    that $a \in L_i$ for some $i$.  If $c$ belongs to $L_i$, then $b$
    must be in $L_i$ as well and\newekki{,} hence\newekki{,} $bc \in
    E^+$ due to operation~(1).  If both $b$ and $c$ are in $L_{i -
      1}$, then again $bc \in E^+$ due to operation~(1). Consider now
    the remaining \newekki{case}, i.e., $a,b \in L_i$ and $c \in L_{i
      - 1}$.  If $ac \in E$ then $bc \in E^+$ because \newekki{either
      $ab \in E$ and thus operation~(2) applies, or $ab \notin E$ and
      thus Claim~2 implies $bc \in E^+$.}  If $ac \in E^+ \setminus E$
    then, according to operation~(2), edge $ac$ was created in $G^+$
    because some vertex $a'\in L_i$ existed such that $\sigma(a') <
    \sigma(a)$ and $a'a, a'c \in E$.  Since $a'c \in E$ and
    $\sigma(a') < \sigma(b) < \sigma(c)$, as before, $bc \in E^+$ must
    hold.~\qedc
  \end{proof}

  \medskip

  To complete the proof of Proposition~\ref{prop:AT}, we recall that a
  graph is an interval graph if and only if it has a
  path-decomposition with each bag being a maximal clique (see,
  e.g.,~\cite{Diestel00,FulGro1965,GilHof64,Gol-book}).  Furthermore,
  such a path\newekki{-}decomposition of an interval graph can
  \newekki{easily} be computed in linear time.  Let $\cP(G^+) = \{X_1,
  X_2, \dots, X_q\}$ be a path-decomposition of our interval graph
  $G^+$.  Then, $\cP(G) := \cP(G^+) = \{X_1, X_2, \ldots, X_q\}$ is a
  path decomposition of $G$ with length at most $2$ since, for every
  edge $uv$ of $G^+$, the distance in $G$ between $u$ and $v$ is at
  most $2$\newekki{, as shown in Claim~3}.~\qed
\end{proof}

Algorithm~\ref{algo:ATfreeDeco} formalizes the steps described in the
previous proof.

\begin{algorithm}
  \KwIn{An AT-free graph $G=(V,E)$.}%
  \KwOut{A path-decomposition of $G$.}%

  \smallskip

  Calculate a\newekki{n} LBFS$(s,G)$ ordering $\sigma$ with an
  arbitrary start vertex $s \in V$.  Let $x$ be the last visited
  vertex, i.e., $\sigma(x) = 1$.

  Calculate a\newekki{n} LBFS$(x,G)$ ordering $\sigma'$.

  Set $E^+ := E$.

  \ForEach {vertex pair $u, v$ with $d_G(x,u) = d_G(x,v)$ and
    $\sigma'(u) < \sigma'(v)$}%
  {%
    Add $uv$ to $E^+$.

    For each $w \in N_G(u)$ with $\sigma'(v) < \sigma'(w)$, add $vw$
    to $E^+$.%
  }

  Calculate a path-decomposition $\cP(G^+)$ of the interval graph $G^+
  = (V, E^+)$ by determining the maximal cliques of $G^+$.

  Output $\cP(G^+)$.

  \caption{\newekki{Computing} a path-decomposition of length at most
    $2$ for a given AT-free graph.}\label{algo:ATfreeDeco}
\end{algorithm}

As the class of cocomparability graphs is a proper subclass of AT-free
graphs, we obtain \newekki{the following corollary}.

\begin{corollary}\label{prop:cocomp}
  If $G$ is an $n$-vertex cocomparability graph, then $\pb(G) \leq
  \pl(G) \leq 2$.  Furthermore, a path-decomposition of $G$ with
  length at most $2$ can be computed in $\cO(n^2)$ time.
\end{corollary}

The complement of an induced cycle on six vertices shows that the
bound $2$ on the path-breadth of cocomparability \newekki{graphs} (and
therefore, of AT-free \newekki{graphs}) is sharp.  Indeed, the edge
set of $\overline{C}_6$ forms a bramble but no vertex $1$-dominates
all edges, implying, by Proposition~\ref{prop:bramble}, that
$\pb(\overline{C}_6) = 2$.


Since the minimum line-distortion problem is NP-hard on
cocomparability graphs~\cite{HM10}, it is NP-hard also on bounded
path-length graphs.

\begin{corollary}\label{cor:ldNPpl}
  The minimum line-distortion problem is NP-hard on bounded
  path-length graphs.
\end{corollary}

We know that the minimum bandwidth problem is NP-hard even on bounded
path-width graphs (e.g., even on caterpillars of hair-length at most
$3$~\cite{Monien86,DubeyaFeige2011}).  Recently, in~\cite{ShTaUe2012},
it was shown that the minimum bandwidth problem is NP-hard also on
so-called {\em convex bipartite graphs}.  A bipartite graph $G = (U,
V; E)$ is said to be {\em convex} if \newekki{for} one of its parts,
say $U$, \newekki{there is an ordering} $(u_1, u_2, \ldots, u_q)$ such
that for all $v \in V$ the vertices adjacent to $v$ are consecutive.
It is easy to see
that, in this case, $\{N_G[u_1], N_G[u_2], \ldots, N_G[u_q]\}$ is a
path-decomposition of $G$ of breadth $1$.  Note that, given a convex
bipartite graph $G = (U, V; E)$, a proper ordering $(u_1, u_2, \ldots,
u_q)$ of $U$ can be found in linear time~\cite{BoothLueker}.  Thus,
the following two results are true.

\begin{proposition}\label{prop:ChBip}
  If $G$ is a convex bipartite graph, then $\pb(G) = 1$ and,
  therefore, $\pl(G) \leq 2$.  Furthermore, a path-decomposition of
  $G$ with breadth $1$ can be computed in linear time.
\end{proposition}

\begin{corollary}\label{cor:bwNPpl}
  The minimum bandwidth problem is NP-hard on bounded path-length
  graphs.
\end{corollary}

\section{Constant-factor approximation of line-distortions of AT-free
  graphs} \label{sec:appr-ld-AT}

From Theorem \ref{th:ld-pl} and results of the previous section, it
follows already that there is an efficient constant-factor
approximation algorithm for the minimum line-distortion problem on
such particular graph classes as permutation graphs, trapezoid graphs,
cocomparability graphs as well as AT-free graphs.  Recall that for
arbitrary (unweighted) graphs the minimum line-distortion problem is
hard to approximate within a constant factor~\cite{BDG+05}.
Furthermore, the problem remains NP-hard even when the input graph is
restricted to a
chordal, cocomparability, or AT-free graph~\cite{HM10}.
Polynomial-time constant-factor approximation algorithms were known
only for split and cocomparability graphs;~\cite{HM10} gave efficient
$6$-approximation algorithms for both graph classes.  As far as we
know, for AT-free graphs (the class which contains all cocomparability
graphs), no prior efficient approximation algorithm was known.

In this section, using additional structural properties of AT-free
graphs and ideas \newekki{from} Section~\ref{ssec:pb-vs-ld}, we give a
better approximation algorithm for all AT-free graphs\newekki{; more
  precisely, we give} an $8$-approximation algorithm \newekki{that}
runs in linear time.

The following nice structural result from~\cite{KKM99} will be very
useful.

\begin{lemma}[\cite{KKM99}]\label{lem:ATfreeLayering}
  Let $G = (V, E)$ be an AT-free graph.  Then, there is a dominating
  path $\pi = (v_0, \ldots, v_k)$ and a layering $\mathcal{L} = \{L_0,
  \ldots, L_k\}$ with $L_i = \{ u \in V : d_G(u, v_0) = i\}$ such that
  for all $u \in L_i$ $(i \geq 1)$, $uv_i \in E$ or $uv_{i-1} \in E$.
  Computing $\pi$ and $\mathcal{L}$ can be done in linear time.
\end{lemma}



\begin{theorem}\label{theo:ATfreeLineApprox}
  There is a linear time algorithm to compute an $8$-approximation of
  the line-distortion of an AT-free graph.
\end{theorem}

\begin{proof}
  Let $G$ be an AT-free graph.  We first compute a path~$\pi = (v_0,
  \ldots, v_k)$ and a layering~$\mathcal{L} = \{L_0, \ldots, L_k\}$ as
  defined in Lemma~\ref{lem:ATfreeLayering}.  To define an embedding
  $f$ of $G$ into the line, we partition every layer~$L_i$ in three
  sets: $\{v_i\}$, $X_i = \{ x : x \in L_i, v_ix \in E \}$, and
  $\overline{X}_i = L_i \setminus (\{v_i\} \cup X_i)$
  (see~Fig.~\ref{fig:ATfreeLineApprox}).  Note that if $x \in
  \overline{X}_i$, then $v_{i - 1}x \in E$.  Since each vertex in
  $X_i$ is adjacent to $v_i$ and each vertex in $\overline{X}_i$ is
  adjacent to $v_{i - 1}$, for all $x,y \in X_i$, $d_G(x, y) \leq 2$,
  and for all $x, y \in \overline{X}_i$, $d_G(x, y) \leq 2$.  Also,
  for all $x \in X_i$ and $y \in \overline{X}_i$, $d_G(x, y) \leq 3$.
  The embedding $f$ places vertices of $G$ into the line in the
  following order: $(v_0, \ldots, v_{i-1}, \overline{X}_{i}, X_{i},
  v_i, \overline{X}_{i + 1}, X_{i + 1}, v_{i+1}, \ldots, v_k)$.
  Between every two vertices $x, y$ placed next to each other
  \newekki{on} the line, to guarantee non-contractiveness, $f$ leaves
  a space of length $d_G(x, y)$ (which is either $1$ or $2$ or $3$,
  where $3$ occurs only when $x \in \overline{X}_i$ and $y \in {X}_i$
  for some $i$).


  \begin{figure}[htb]
    \begin{center} 
      \includegraphics[height=5.0cm]{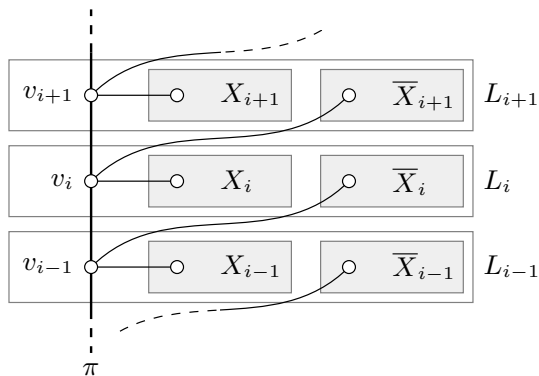}
      \caption{\label{fig:ATfreeLineApprox} Layering of an AT-free
        graph; illustration to the proof of
        Theorem~\ref{theo:ATfreeLineApprox}.  } %
    \end{center}
    \vspace*{-5mm}
  \end{figure}

  \newekki{Now } we will show that $f$ approximates the minimum
  line-distortion of $G$.  Since $\mathcal{L}$ is a BFS layering
  \newekki{started} from $v_0$, i.e., it represents the distances of
  vertices from $v_0$, there is no edge~$xy$ with $x \in L_{i-1}$ and
  $y \in L_{i+1}$.  Also note that $D_G(v_i, 2) \supseteq L_i \cup
  L_{i + 1} \cup \{v_{i - 1}\}$.  By the definition of $f$, for all
  $xy \in E$ with $x, y \in L_{i} \cup L_{i + 1}$, $|f(x) - f(y)| <
  |f(v_{i - 1}) - f(v_{i + 1})|$.  Therefore, counting how many
  vertices are placed by $f$ between $f(v_{i - 1})$ and $f(v_{i + 1})$
  and the distance in $G$ between vertices placed next to each other,
  we get $|f(x) - f(y)| \leq 2(|D_G(v_i, 2)| - 2) + 2 = 2(|D_G(v_i,
  2)| - 1)$.  Using 
  Lemma~\ref{lem:Discbandwidth} and the fact that $\bw(G) \leq \ld(G)$,
  we get $|f(x) - f(y)| \leq 8 \, \ld(G)$ for all $xy \in E$. \qed
\end{proof}

Algorithm~\ref{algo:ATfreeLineDesApprox} formalizes the method
described above.

\begin{algorithm}
  \KwIn{An AT-free graph $G=(V,E)$.}%
  \KwOut{ An embedding $f$ of $G$ into the line.}%

  \smallskip

  Compute a path~$\pi = (v_0, \ldots, v_k)$ and a
  layering~$\mathcal{L} = \{L_0, \ldots, L_k\}$ as defined in
  Lemma~\ref{lem:ATfreeLayering}.

  Partition \newekki{each} layer~$L_i$ into three sets: $\{v_i\}$,
  $X_i = \{ x : x \in L_i, v_ix \in E \}$, and $\overline{X}_i = L_i
  \setminus (\{v_i\} \cup X_i)$.

  Create an embedding $f$ by placing the vertices of $G$ into the line
  in the order $(v_0, \ldots, v_{i-1}, \overline{X}_{i}, X_{i}, v_i,
  \overline{X}_{i + 1}, X_{i + 1}, v_{i + 1}, \ldots, v_k)$.

  Between every two \newekki{consecutive} vertices $x,y$ \newekki{on}
  the line, leave a space of length $d_G(x, y)$.

  Output $f$.

  \caption{An $8$-approximation algorithm for the minimum
    line-distortion of an AT-free
    graph.}\label{algo:ATfreeLineDesApprox}
\end{algorithm}

It is easy to see that the order in which \newekki{the} vertices of
$G$ \newekki{are} placed by $f$ into the line gives also a layout of
$G$ with bandwidth at most $4 \, \bw(G)$.  This reproduces an
approximation result from~\cite{KKM99} (in fact, their algorithm
\newekki{has} complexity $\cO(m + n \log n)$ for an $n$-vertex
$m$-edge graph, since it \newekki{involves} an $\cO(n \log n)$ time
algorithm from~\cite{APSZ81} to find an optimal layout for a
caterpillar with hair-length at most $1$).

\begin{corollary}[\cite{KKM99}]\label{cor:ATfreeBWApprox}
  There is a linear time algorithm to compute a $4$-approximation of
  the minimum bandwidth of an AT-free graph.
\end{corollary}

Combining Proposition~\ref{prop:AT} and Proposition~\ref{prop:PD-bw},
we obtain also the following result from~\cite{KKM99} as a corollary.

\begin{corollary}[\cite{KKM99}]\label{cor:bw-AT-2-appr}
  There is an $\cO(m + n \log^2 n)$ time algorithm to compute a
  $2$-approximation of the minimum bandwidth of an AT-free graph.
\end{corollary}


\section{Concluding remarks}

\newekki{In this paper} we have shown that if a graph $G$ has a
$k$-dominating shortest path\newekki{, where} $k$ is a
constant\newekki{,} or the path-length of $G$ is bounded by a
constant, then both the minimum line-distortion problem and the
minimum bandwidth problem on $G$ can \newekki{be} efficiently
approximated within constant factors.  As AT-free graphs,
cocomparability graphs, permutation graphs, trapezoid graphs, convex
bipartite graphs, caterpillars with hairs of bounded length, all have
bounded path-length or have $k$-dominating shortest paths with
constant $k$, they admit constant factor approximations of the minimum
bandwidth and the minimum line-distortion.  Thus, \newekki{the}
constant factor approximation results of~\cite{HM10,KKM99,ShTaUe2012}
become special cases of our results.

We conclude this paper with a few open questions.  We have presented a
$2$-approximation algorithm for computing the path-length of a general
graph but we do not know the complexity status of this problem.
So, our first open question is the following.
\begin{enumerate}[1)]
\item Is \newekki{it NP-complete to decide} whether a graph has
  path-length at most $k$ ($k > 1$)?
\end{enumerate}
We gave a first constant-factor approximation ($8$-approximation)
algorithm for the minimum line-distortion problem on AT-free graphs
and reproduced a $4$-approximation and a $2$-approximation for the
minimum bandwidth on such graphs. 
\begin{enumerate}
\item [2)] Does there exist a better approximation algorithm for the
  minimum line-distortion problem on AT-free graphs?
\end{enumerate}
We have mentioned that the minimum bandwidth problem is notoriously
hard 
even on bounded path-width graphs (e.g., even on caterpillars of
hair-length at most $3$~\cite{Monien86,DubeyaFeige2011}) and even on
bounded path-length graphs (e.g., even on convex bipartite
graphs~\cite{ShTaUe2012}). 
Since the minimum line-distortion problem is NP-hard on
cocomparability graphs~\cite{HM10}, it is NP-hard also on bounded
path-length graphs.  However, the status of the minimum
line-distortion problem on bounded path-width graphs is unknown.  See
Table~\ref{tbl:NPhardBwLd} for a summary.


\begin{enumerate}
\item [3)] Is the minimum line-distortion problem on bounded
  path-width graphs NP-hard?
\end{enumerate}
We are \newekki{also} interested in a more general
question.  
\begin{enumerate}
\item [4)] Is there a better hardness result (\newekki{better} than a
  constant~\cite{BDG+05}) for the line-distortion problem in general
  graphs?
\end{enumerate} 

\begin{table}
\scriptsize
  \centering
  \caption{Hardness results for the minimum line-distortion problem
    and the minimum bandwidth problem on graphs with bounded path-with
    or bounded path-length.}\label{tbl:NPhardBwLd}
  \begin{tabular}{l     c   c}
    \toprule
    & $\pw(G) \leq c$  & $\pl(G) \leq c$ \\
    \midrule
    bandwidth & ~NP-hard (caterpillars with hair-length at most $3$ \cite{Monien86})~ & ~NP-hard (convex bipartite graphs \cite{ShTaUe2012})~ \\
    line-distortion~ & ? & ~NP-hard (cocomparability graphs \cite{HM10}) \\
    \bottomrule
\end{tabular}
\end{table}

\end{document}